\documentclass[journal]{IEEEtran}
\usepackage{amsmath,graphicx,amssymb,amsthm}
\usepackage{algorithm}
\usepackage{algorithmic}
\newtheorem{theorem}{Theorem}
\newtheorem*{remark}{Remark}
\newtheorem{lemma}{Lemma}

\usepackage{cite}
\usepackage{color}
\def\blue{\color{black}}
\def\red{\color{black}}
\usepackage{hyperref}

\title{How to Mobilize mmWave: A Joint Beam and Channel Tracking Approach}
\author{Jiahui Li$^\dagger$, Yin Sun$^\S$, Limin Xiao$^{\ddagger\P}$, Shidong Zhou$^\dagger$, Ashutosh Sabharwal$^*$ \\
$^\dagger$Dept. of EE, $^\ddagger$Research Institute of Information Technology, Tsinghua University, Beijing, 100084, China\\
$^\S$Dept. of ECE, Auburn University, Auburn AL, 36849, U.S.A.\\
$^*$Dept. of ECE, Rice University, Houston TX, 77251, U.S.A.\\
\thanks{$^\P$Corresponding author.}
\thanks{J. Li, L. Xiao, and S. Zhou were supported in part by National S\&T Major Project grant 2017ZX03001011-002, National Natural Science Foundation of China grant 61631013, National High Technology Research and Development Program of China (863 Program) grant 2014AA01A703, Science Fund for Creative Research Groups of NSFC grant 61621091, Tsinghua University Initiative Scientific Research grant 2016ZH02-3, International Science and Technology Cooperation Program grant 2014DFT10320, ”2011 Plan” Wireless Communication Technology Co-Innovation Center grant 20161210020, Tsinghua-Qualcomm Joint Research Program. Y. Sun was supported in part by ONR grant N00014-17-1-2417. A. Sabharwal was supported in part by NSF grants CNS-1518916 and CNS-1314822.}
}

\begin{document}
%\ninept
%
\maketitle
\begin{abstract}
Maintaining reliable millimeter wave (mmWave) connections to many fast-moving mobiles is a key challenge in the theory and practice of 5G systems. In this paper, we develop a new algorithm that can jointly track {\blue the beam direction and channel coefficient of mmWave propagation paths using phased antenna arrays}. Despite the significant difficulty in this problem, our algorithm can simultaneously achieve fast tracking speed, high tracking accuracy, and low pilot overhead. In static scenarios, this algorithm can converge to the minimum Cram\'er-Rao lower bound of {\blue beam direction} with high probability. {\blue Simulations reveal that this algorithm greatly outperforms several existing algorithms.} Even at SNRs as low as 5dB, our algorithm is capable of tracking a mobile moving at an angular velocity of 5.45 degrees per second and achieving over 95\% of channel capacity with a 32-antenna phased array, by inserting only 10 pilots per second.
%\vspace{-3mm}
\end{abstract}
%\begin{keywords}
%Beam and channel tracking, fast tracking speed, high accuracy, mmWave, phased antenna arrays.
%\end{keywords}

\section{Introduction}
\label{sec:intro}

Millimeter-wave (mmWave) communication is promising to support the vastly growing data traffic for future wireless systems \cite{Pi2011An, Boccardi2014Five, Heath2016overview}. In the mmWave band, only several distinctive propagation paths exist, i.e., the line-of-sight path and a few relatively strong reflected paths \cite{Rappaport2013Millimeter, Rappaport2015Wideband}. {\blue Therefore,} the directional beamforming with large antenna arrays is necessary to provide sufficiently strong received signal power. 

To overcome the hardware limitation {\blue on the number of radio frequency (RF) chains} with large array size and high carrier frequency, analog beamforming with phased antenna arrays was proposed \cite{Sun2014Mimo, Han2015Large, Puglielli2016Design, Molisch2016Hybrid, Heath2016overview}. {\blue A phased array can receive the signal that is projected onto a certain spatial subspace, with a cost of requiring much more pilots than the fully digital arrays to find the rare and precious paths.} When users move quickly, it is needed to track the dynamic paths and even more pilots are required. % Otherwise, the directional beams cannot be steered to the right beam directions and issues will happen. 
Hence, one fundamental challenge is how to accurately track a large number of dynamic paths from many high-mobility terminals/reflectors using limited pilots, {\blue e.g., in V2V/V2I, high-speed railway, and UAV scenarios} \cite{Brown2016Promise}. 

The compressed sensing based algorithms (e.g., \cite{Wang2009Beam, Alkhateeb2014Channel, Alkhateeb2015Compressed}) were proposed for phased arrays, which can reduce pilot overhead and make beam direction acquisition faster. However, these algorithms are designed for static or quasi-static scenarios, and will encounter performance deterioration under high-mobility scenarios. To cope with high-mobility scenarios, the algorithms in \cite{IEEE80211ad, palacios2016tracking, Gao2016Fast} use the prior information to track the dynamic beam directions. However, these solutions do not optimize the tracking scheme with the optimal training beamforming vectors, which leads to poor tracking accuracy. 

{\blue Since the tracking of a large number of dynamic paths can be decoupled into tracking each path with low pilot overhead, we have proposed a beam tracking algorithm in \cite{Li2017conf, Li2017super} to optimize both the training beamforming vectors and tracking scheme.} However, it assumes {\blue known} channel coefficients, while both channel coefficient and beam direction might be {\blue unknown} and time-varying in a real mobile system. In this paper, we further develop a recursive beam and channel tracking (RBCT) algorithm to jointly track the dynamic beam direction and channel coefficient. 
%Contrary to prior works \cite{Wang2009Beam, Alkhateeb2014Channel, Alkhateeb2015Compressed, IEEE80211ad, palacios2016tracking, Gao2016Fast, Li2017conf, Li2017super}, it aims to optimize the training beamforming vectors, as well as the beam direction and channel coefficient tracker. 
In static scenarios, the Cram\'er-Rao lower bound (CRLB) of beam direction is derived, which is a function of the training beamforming vectors. We also obtain the \emph{minimum} CRLB by optimizing these {\blue training} beamforming vectors, and establish three theorems to verify that the RBCT algorithm can converge to the \emph{minimum} CRLB with high probability. Simulations reveal that the RBCT algorithm can achieve much faster tracking speed, lower tracking error, and lower pilot overhead than several existing algorithms.

We use the following notations: $\mathbf{A}$ is a matrix, $\mathbf{a}$ is a vector, $a$ is a scalar. $\left\| \mathbf{A} \right\|_2$ is the 2-norm of $\mathbf{A}$. $\mathbf{A}^\text{T}$, $\mathbf{A}^\text{H}$ and $\mathbf{A}^{-1}$ are $\mathbf{A}$'s transpose, Hermitian and inverse, respectively. $\mathbb{E}[\cdot]$ denotes expectation and $\operatorname{Re}\left\{\cdot\right\}$($\operatorname{Im}\left\{\cdot\right\}$) obtains the real (imaginary) part. The natural logarithm of $x$ is $\log(x)$.
%\vspace{-1.5mm}
\section{System Model}
%\vspace{-0.5mm}
\label{sec:model}

Consider a phased array in Fig. \ref{fig:system}, where $M$ {\blue omnidirectional} antennas are placed on a line, with a distance $d$ between two neighboring antennas. {\blue Each antenna is connected through a phase shifter to the same RF chain.} In time-slot $n$, the pilot symbols arrive at the array from an angle-of-arrival (AoA) $\theta_n\!\in\![-\frac{\pi}{2},\frac{\pi}{2}]$. The channel vector is given by
%\vspace{-2.5mm}
\begin{equation}\label{eq:channel}
\mathbf{h}_n = \beta_n \mathbf{a}(x_n),
%\vspace{-2.5mm}
\end{equation}
where $x_n\!=\!\sin(\theta_n)$ is the sine of the AoA $\theta_n$, $\mathbf{a}(x_n)\!\!=\!\!\left[ 1,\!e^{j \frac{2\pi d}{\lambda} x_n},\!\cdots,\!e^{j \frac{2\pi d}{\lambda}(M\!-\!1)x_n} \right]^\text{H}$ is the steering vector of the arriving beam, $\lambda$ is the wavelength, and $\beta_n\!=\!\beta^\text{re}_n\!+\!j\beta^\text{im}_n$ is the complex channel coefficient.
%with $\beta^{\text{re}}_n \overset{\Delta}{=}\operatorname{Re}\left\{\beta_n\right\}$, $\beta^{\text{im}}_n\overset{\Delta}{=}\operatorname{Im}\left\{\beta_n\right\}$.

\begin{figure}
\centering
%%\vspace{-4mm}
\includegraphics[width=8cm]{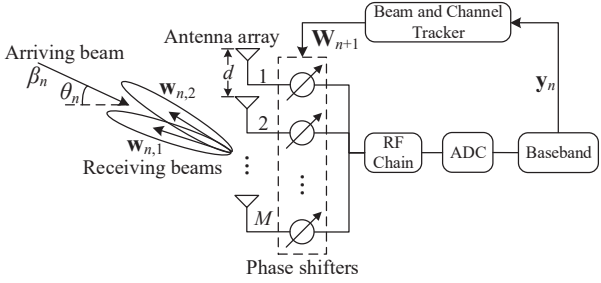}
%\vspace{-4mm}
\caption{System model.}
%\vspace{-6mm}
\label{fig:system}
\end{figure}

To track the beam direction $x_n$ and channel coefficient $\beta_n$ simultaneously, at least two observations using different beamforming vectors are needed. Hence, we assume that two pilot symbols are {\blue applied} in each time-slot. To receive the $i$-th ($i = 1, 2$) pilot symbol, let $\mathbf{w}_{n,i}$ be the \textbf{beamforming vector} in time-slot $n$, denoted by
\vspace{-3mm}
\begin{equation}\label{eq:bf}
	\mathbf{w}_{n,i} = \frac{\mathbf{a}(x_n+\delta_{n, i})}{\sqrt{M}},
\vspace{-2.5mm}
\end{equation}
%let $w_{mn,i} \in[-\pi,\pi]$ be the phase shift in radians provided by the $m$-th phase shifter in time-slot $n$. Then, the beamforming vector is denoted by
%\vspace{-2mm}
%\begin{equation}\label{eq:bf}
%	\mathbf{w}_{n,i} = \frac{1}{\sqrt{M}}\left[ e^{jw_{1n,i}}, \cdots, e^{jw_{Mn,i}} \right]^\text{H} = \frac{\mathbf{a}(x_n+\delta_{n, i})}{\sqrt{M}},
%\vspace{-2mm}
%\end{equation}
which is assumed to have the same form as the steering vector. Combining the output signals of the phase shifters yields
\vspace{-2mm}
\begin{equation}\label{eq:observation}
\begin{aligned}
y_{n,i} = \mathbf{w}_{n,i}^\text{H}\mathbf{h}_n{s} + {z}_{n,i} = \beta_{n} \mathbf{w}_{n,i}^\text{H}\mathbf{a}(x_n){s} + {z}_{n,i},
\end{aligned}
\vspace{-2mm}
\end{equation}
where ${s}$ is the pilot symbol that is known by the receiver, and ${z}_{n,i}\!\sim\!\mathcal{CN}(0,\sigma_0^2)$ is an \emph{i.i.d.} circularly symmetric complex Gaussian random variable. Given $\boldsymbol{\psi}_n\!=\![\beta^\text{re}_n, \beta^\text{im}_n, x_n]^\text{T}$ and $\mathbf{W}_{n}\!=\!\left[\mathbf{w}_{n,1},\mathbf{w}_{n,2}\right]$, the conditional probability density function of $\mathbf{y}_n\!=\![y_{n,1}, y_{n,2}]^\text{T}$ is given by
\vspace{-3mm}
\begin{equation}
\label{eq:pdf}
\begin{aligned}
	p(\mathbf{y}_n| \boldsymbol{\psi}_n, \mathbf{W}_n) = \frac{1}{\pi^2\sigma_0^4} e^{- \frac{\left\| \mathbf{y}_n -  {s} \beta_{n} \mathbf{W}_{n}^\text{H}\mathbf{a}(x_n) \right\|^2_2}{\sigma_0^2}}.
\end{aligned}
\vspace{-2mm}
\end{equation}
A beam and channel tracker determines the \textbf{beamforming matrix} $\mathbf{W}_{n}$, and provides an estimate $\hat{\boldsymbol{\psi}}_n\!=\![\hat{\beta}^\text{re}_n, \hat{\beta}^\text{im}_n, \hat{x}_n]^\text{T}$ of the channel coefficient $\beta_n$ and the sine $x_n$ of the AoA. Let $\xi\!=\!(\mathbf{W}_1, \mathbf{W}_2, \ldots, \hat{\boldsymbol{\psi}}_1, \hat{\boldsymbol{\psi}}_2, \ldots)$ be a \textbf{beam and channel tracking policy}. In particular, we consider the set $\Xi$ of \emph{causal} beam and channel tracking policies: The estimate $\hat{\boldsymbol{\psi}}_n$ of time-slot $n$ and the beamforming matrix $\mathbf{W}_{n+1}$ of time-slot $n+1$ are determined  by using the history of beamforming matrices $(\mathbf{W}_1, \ldots, \mathbf{W}_n)$ and channel observations $(\mathbf{y}_1, \ldots, \mathbf{y}_n)$. 

\section{Joint Beam and Channel Tracking Problem}
\label{sec:problem}
%In Section \ref{ssec:formulation}, we first formulate the beam tracking problem. Then, in Section \ref{ssec:bound}, we study a fundamental performance bound for the beam tracking problem.

Our goal is to develop a joint beam and channel tracking algorithm to minimize the beam tracking error. For any time-slot $n$, the joint beam and channel tracking problem is given by
%\vspace{-2mm}
\begin{equation}\label{eq:problem}
\begin{aligned}
	\underset{\begin{matrix}\xi \in \Xi \end{matrix}}{\min}~& \mathbb{E}\left[ \left( \hat{x}_{n} - x_n \right)^2 \right] \\
	\text{s.t.}~&  \mathbb{E}\left[ \hat{\beta}_{n} \right] = \beta_n,~\mathbb{E}\left[ \hat{x}_{n} \right] = x_n,
\end{aligned}
%\vspace{-2mm}
\end{equation}
where the constraint ensures that $\hat{\boldsymbol{\psi}}_n\!=\![\hat{\beta}^\text{re}_n, \hat{\beta}^\text{im}_n, \hat{x}_n]^\text{T}$ is an \emph{un-biased} estimate of $\boldsymbol{\psi}_n\!=\![\beta^\text{re}_n, \beta^\text{im}_n, x_n]^\text{T}$. 

Problem \eqref{eq:problem} is a constrained sequential control and estimation problem that is difficult to solve optimally, where the beamforming matrix $\mathbf{W}_{n}$ is the control action. {\red First, the system is only partially observed through the channel observation $\mathbf{y}_n$. Second, both the beamforming matrix $\mathbf{W}_{n}$ and the estimate $\hat{\boldsymbol{\psi}}_n$ need to be optimized: On the one hand, the optimization of $\mathbf{W}_{n}$ is a non-convex optimization problem of $\delta_{n,i}$ in \eqref{eq:bf}, which is discussed in Section \ref{ssec:bound}. On the other hand, as will be discussed in Section \ref{sec:analysis}, the optimization of $\hat{\boldsymbol{\psi}}_n$ is also non-convex and has multiple local optimal estimates.}

%%\vspace{-2mm}
\subsection{Cram\'er  Rao Lower Bound of Beam Tracking}\label{ssec:bound}
%%\vspace{-2mm}

\newcounter{mytempeqncnt}
\begin{figure*}[!t]
\normalsize
\setcounter{mytempeqncnt}{\value{equation}}
\setcounter{equation}{12}
%\vspace{-5mm}
\begin{equation}
\label{eq:newton2}
\begin{aligned} 
\!\!&\hat{\boldsymbol{\psi}}_n\!=\!\hat{\boldsymbol{\psi}}_{n\!-\!1}\!-\!\frac{a_n}{\left\|{s}\hat{\mathbf{g}}_{n}\right\|^2_2 (l_n^2 \!-\! |c_n|^2)}\!\cdot\!\left[\!\begin{matrix} l_n^2\!-\!\operatorname{Im}\{c_n\}^2  \!&\! \operatorname{Re}\{c_n\}\operatorname{Im}\{c_n\} \!&\! -\|\hat{\mathbf{g}}_{n}\|^2_2\operatorname{Re}\{c_n\} \\ 
\operatorname{Re}\{c_n\}\operatorname{Im}\{c_n\} \!&\! l_n^2\!-\!\operatorname{Re}\{c_n\}^2  \!&\! -\|\hat{\mathbf{g}}_{n}\|^2_2\operatorname{Im}\{c_n\} \\
-\|\hat{\mathbf{g}}_{n}\|^2_2\operatorname{Re}\{c_n\} \!&\! -\|\hat{\mathbf{g}}_{n}\|^2_2\operatorname{Im}\{c_n\} \!&\! \|\hat{\mathbf{g}}_{n}\|^4_2 \end{matrix}\!\right]\!\cdot\!\left[\begin{matrix} \operatorname{Re}\{{s}^\text{H}\hat{\mathbf{g}}_{n}^\text{H}(\mathbf{y_n}\!-\!{s}\hat{\beta}_{n\!-\!1}\hat{\mathbf{g}}_{n})\} \\
\operatorname{Im}\{{s}^\text{H}\hat{\mathbf{g}}_{n}^\text{H}(\mathbf{y_n}\!-\!{s}\hat{\beta}_{n\!-\!1}\hat{\mathbf{g}}_{n})\} \\
\operatorname{Re}\{{s}^\text{H}\hat{\mathbf{e}}_{n}^\text{H}(\mathbf{y_n}\!-\!{s}\hat{\beta}_{n\!-\!1}\hat{\mathbf{g}}_{n})\} \end{matrix}\right]\!\!.\!\!\!\!\!\!\!\!
\end{aligned}
%\vspace{-1.5mm}
\end{equation}
\setcounter{equation}{\value{mytempeqncnt}}
\hrulefill
%\vspace{-5mm}
\end{figure*}

Now, we try to establish a lower bound of the MSE in \eqref{eq:problem} in \emph{static} scenarios, {\red where the ground true of beam direction and channel coefficient is invariant for all time-slot $n$, i.e., $\boldsymbol{\psi}_n\!=\![\beta^\text{re},\!\beta^\text{im},\!x]^\text{T}\!\overset{\Delta}{=}\!\boldsymbol{\psi}$. Given the beamforming matrices $(\mathbf{W}_1, \ldots, \mathbf{W}_n)$ of the first $n$ time-slots, the MSE in \eqref{eq:problem} is lower bounded by the CRLB as follows \cite{nevel1973stochastic}:}
%\vspace{-1.5mm}
\begin{equation}\label{eq:MMSE}
	\begin{aligned}
	\mathbb{E}\left[ \left( \hat{x}_{n} - x \right)^2 \right] \ge &~ \left[ \left(\sum_{i=1}^n \mathbf{I}(\boldsymbol{\psi}, \mathbf{W}_i)\right)^{-1} \right]_{3,3},
	% \ge &~\frac{1}{n}\left[ \mathbf{I}(\boldsymbol{\psi}, \mathbf{W}^*)^{-1} \right]_{3,3} = \frac{1}{nI_{\max}},
	\end{aligned}
	%\vspace{-1.5mm}
\end{equation}
where $[\cdot]_{i,k}$ obtains the matrix element in row $i$ and column $k$, and $\mathbf{I}(\boldsymbol{\psi}, \mathbf{W}_i)$ is the $3\times 3$ Fisher information matrix, i.e., \cite{Poor1994estimation}
%\vspace{-3mm}
\begin{equation}\label{eq:FIM}
\begin{aligned}
\!\!\!\mathbf{I}(\boldsymbol{\psi}, \mathbf{W}_i)&\!\overset{\Delta}{=}\!\mathbb{E}\!\left[\frac{\partial \log p(\mathbf{y}_i| \boldsymbol{\psi},\!\mathbf{W}_i)}{\partial \boldsymbol{\psi}}\!\cdot\!\frac{\partial \log p(\mathbf{y}_i| \boldsymbol{\psi},\!\mathbf{W}_i)}{\partial \boldsymbol{\psi}^\text{T}}\right]\!\!\! \\
&\!\!\!\!\!\!\!\!\!\!\!\!\!\!\!\!= \frac{2|{s}|^2}{\sigma_0^2}\!\! \left[\begin{matrix} \left\|\mathbf{g}_i\right\|_2^2 & 0 & \operatorname{Re}\left\{\mathbf{g}_i^\text{H}\mathbf{e}_i\right\} \\ 
0 & \left\|\mathbf{g}_i\right\|_2^2 & \operatorname{Im}\left\{\mathbf{g}_i^\text{H}\mathbf{e}_i\right\} \\
\operatorname{Re}\left\{\mathbf{g}_i^\text{H}\mathbf{e}_i\right\} & \operatorname{Im}\left\{\mathbf{g}_i^\text{H}\mathbf{e}_i\right\} & \left\|\mathbf{e}_i \right\|_2^2 \end{matrix}\right], 
\end{aligned}
%\vspace{-2mm}
\end{equation}
where $\mathbf{g}_i \!=\! \mathbf{W}_i^\text{H}\mathbf{a}(x)$, $\mathbf{e}_i \!=\! \beta\mathbf{W}_i^\text{H}\dot{\mathbf{a}}(x)$, and $\dot{\mathbf{a}}(x)\!\overset{\Delta}{=}\!\frac{\partial \mathbf{a}(x)}{\partial x}$. {\red By optimizing the beamforming matrices $(\mathbf{W}_1, \ldots, \mathbf{W}_n)$ on the RHS of \eqref{eq:MMSE}, we obtain the \emph{minimum} CRLB as below:
%%\vspace{-2mm}
%\begin{equation}
\begin{align}
\label{eq:CRLB}
\left[ \left(\sum_{i=1}^n \mathbf{I}(\boldsymbol{\psi}, \mathbf{W}_i)\right)^{-1} \right]_{3,3} \!\!\!\!\!&~\ge \min_{\mathbf{W}_1, \ldots, \mathbf{W}_n}\left[ \left(\sum_{i=1}^n \mathbf{I}(\boldsymbol{\psi}, \mathbf{W}_i)\right)^{-1} \right]_{3,3} \nonumber \\
&~ = \min_{\mathbf{W}_i}\frac{1}{n}\left[ \mathbf{I}(\boldsymbol{\psi}, \mathbf{W}_i)^{-1} \right]_{3,3},
%\left[ \big(n \mathbf{I}(\boldsymbol{\psi}, \mathbf{W}^*)\big)^{-1} \right]_{3,3} \\
%&~= \frac{1}{n}\left[ \mathbf{I}(\boldsymbol{\psi}, \mathbf{W}^*)^{-1} \right]_{3,3},
\end{align}
%%\vspace{-2mm}
%\end{equation}
where because the linear additive property of Fisher information matrix \cite{Poor1994estimation}, the optimal $\mathbf{W}_1, \ldots, \mathbf{W}_n$ are the same, and from \eqref{eq:FIM}, we can get
%\vspace{-3.5mm}
\begin{equation}
\left[ \mathbf{I}(\boldsymbol{\psi}, \mathbf{W}_i)^{-1} \right]_{3,3} = \frac{\sigma_0^2}{2|{s} \beta|^2}\cdot\frac{\left\|\mathbf{g}_i\right\|_2^2}{ \left\|\mathbf{g}_i\right\|_2^2 \left\|\mathbf{e}_i\right\|_2^2 - \left| \mathbf{g}_i^\text{H}\mathbf{e}_i \right|^2 }.
%\vspace{-1.5mm}
\end{equation}
Problem \eqref{eq:CRLB} is non-convex with respect to $\delta_{i,1}$ and $\delta_{i,2}$, which makes it too hard to obtain the analytical solution. However, we can still use numerical method to find the solution, which yields the optimal beamforming matrix $\mathbf{W}^*$ as below:
%\vspace{-2mm}
\begin{equation}\label{eq:control_opt}
% \mathbf{w}_1^* = \frac{\mathbf{a}(x-\delta^*)}{\sqrt{M}},~\mathbf{w}_2^* = \frac{\mathbf{a}(x+\delta^*)}{\sqrt{M}},
\mathbf{W}^*\!=\!\frac{1}{\sqrt{M}}\big[\mathbf{a}(x-\delta^*), \mathbf{a}(x+\delta^*)\big],
%\vspace{-2mm}
\end{equation}
where $\delta^* \xrightarrow{M \rightarrow \infty} \frac{2\lambda}{3Md}$, and when $M \ge 8$, $\delta^*$ is very close to $\frac{2\lambda}{3Md}$. In Fig. \ref{fig:optimal_control}, the optimal receiving beam directions are depicted by plotting $\frac{1}{\left[ \mathbf{I}(\boldsymbol{\psi}, \mathbf{W})^{-1} \right]_{3,3}}$  vs. $\delta_{i,1}$ and $\delta_{i,2}$, where $M = 32, d = 0.5\lambda$, and the signal-to-noise ratio (SNR) $\frac{|{s} \beta|^2}{\sigma_0^2}$ is $5\text{dB}$. It can be observed that $\delta^*$ is almost the same as $\frac{2\lambda}{3Md}$ and there are two symmetric optimal solutions. Therefore, we will set $\delta^* = \frac{2\lambda}{3Md}$ in the proposed RBCT algorithm in Section \ref{sec:algorithm}.}
% Moreover, in our further studies, we find that $\frac{3Md}{2\lambda}\cdot\delta^*$ remains almost the same for different antenna numbers (e.g., $M \ge 8$). 

\begin{figure}
\centering
\includegraphics[width=5.5cm]{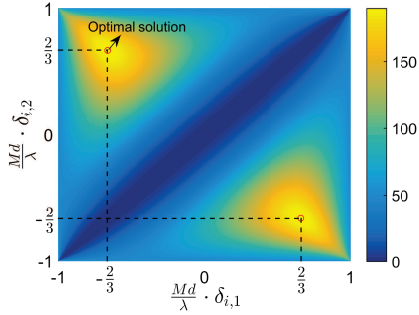}
%\vspace{-3mm}
\caption{Optimization of Problem \eqref{eq:CRLB} using numerical method.}
%\vspace{-2mm}
\label{fig:optimal_control}
\end{figure}

\section{Recursive Beam and Channel Tracking}
\label{sec:algorithm}

\begin{figure}
\centering
\includegraphics[width=6.5cm]{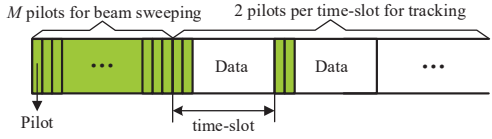}
%\vspace{-3mm}
\caption{Frame structure.}
%\vspace{-6mm}
\label{fig:frame}
\end{figure}

{\blue We propose a two-stage algorithm to approach the minimum CRLB in \eqref{eq:CRLB}, which is given below: }

\noindent\textbf{Recursive Beam and Channel Tracking (RBCT):}

\noindent\textbf{1) Coarse Beam Sweeping:} %Use the exhaustive beam sweeping algorithm \cite{Hur2013Millimeter}.
$M$ pilots are used successively (see Fig. \ref{fig:frame}). The beamforming vector to receive the $m$-th observation $\tilde{y}_m$ is set as $\tilde{\mathbf{w}}_{m}\!=\!\frac{1}{\sqrt{M}}\mathbf{a}\left(\frac{2m}{M} - \frac{M+1}{M}\right), m\!=\!1, \ldots, M$. Obtain the initial estimate $\hat{\boldsymbol{\psi}}_{0}\!=\![\hat{\beta}^\text{re}_{0}, \hat{\beta}^\text{im}_{0}, \hat{x}_{0}]^\text{T}$ by
%\vspace{-2mm}
\begin{equation}\label{eq:initial}
\begin{aligned}
	\hat{x}_{0}\!=\!\underset{\hat{x} \in \mathcal{X}}{\arg\max}\left|\mathbf{a}(\hat{x})^\text{H} \tilde{\mathbf{W}} \tilde{\mathbf{y}}  \right|, \hat{\beta}_{0}\!=\!\left[\tilde{\mathbf{W}}^\text{H} \mathbf{a}(\hat{x}_{0})\right]^+\!\!\tilde{\mathbf{y}},
\end{aligned}
%\vspace{-3mm}
\end{equation}
%\begin{equation}\label{eq:initial}
%\left\{\begin{aligned}
%	& \hat{x}_{0} =\underset{\hat{x} \in \mathcal{X}}{\arg\max}~\left|\mathbf{a}(\hat{x})^\text{H} \tilde{\mathbf{W}} \tilde{\mathbf{y}}  \right| \\
%	& \hat{\beta}_{0} = \left[\tilde{\mathbf{W}}^\text{H} \mathbf{a}(\hat{x}_{0})\right]^+ \tilde{\mathbf{y}},
%\end{aligned}\right.
%\end{equation}
where $\tilde{\mathbf{y}} = [\tilde{y}_1, \ldots, \tilde{y}_M]^\text{T}$, $\tilde{\mathbf{W}} = [\tilde{\mathbf{w}}_1, \ldots, \tilde{\mathbf{w}}_M]$, $\mathcal{X} = \left\{\frac{1 - M_0}{M_0}, \frac{3 - M_0}{M_0}, \ldots, \frac{M_0-1}{M_0}\right\}$,  the size $M_0 (M_0 \ge M)$ of $\mathcal{X}$ determines the estimation resolution, and $\mathbf{X}^+ \overset{\Delta}{=} (\mathbf{X}^\text{H}\mathbf{X})^{-1}\mathbf{X}^\text{H}$.

\noindent\textbf{2) Beam and Channel Tracking:} In time-slot $n$, two pilots are received at the beginning (see Fig. \ref{fig:frame}) using beamforming vectors $\mathbf{w}_{n, 1}$ and $\mathbf{w}_{n, 2}$, given by
%\vspace{-3mm}
\begin{equation}\label{eq:est_ctrl}
	\begin{aligned}
	\mathbf{w}_{n,1}= \frac{\mathbf{a}(\hat{x}_{n-1}-\delta^*)}{\sqrt{M}},~\mathbf{w}_{n,2} = \frac{\mathbf{a}(\hat{x}_{n-1}+\delta^*)}{\sqrt{M}},
	%\left[ 1, e^{j \frac{2\pi d}{\lambda} \hat{x}_{n}}, \cdots, e^{j \frac{2\pi d}{\lambda}(M-1)\hat{x}_{n}} \right]^\text{H}.
	\end{aligned}
	%\vspace{-2mm}
\end{equation}
and  the estimate $\hat{\boldsymbol{\psi}}_n\!=\![\hat{\beta}^\text{re}_n, \hat{\beta}^\text{im}_n, \hat{x}_n]^\text{T}$ is updated by \eqref{eq:newton2} on the top of the page, where $\hat{\mathbf{g}}_{n}\!=\!\mathbf{W}_n^\text{H}\mathbf{a}(\hat{x}_{n\!-\!1})$, $\hat{\mathbf{e}}_{n}\!=\!\hat{\beta}_{n\!-\!1}\mathbf{W}_n^\text{H}\dot{\mathbf{a}}(\hat{x}_{n\!-\!1})$, $l_n = \|\hat{\mathbf{g}}_{n}\|_2 \|\hat{\mathbf{e}}_{n} \|_2$, and $c_n = \hat{\mathbf{g}}_{n}^\text{H}\hat{\mathbf{e}}_{n}$.

\setcounter{equation}{13}

In \emph{Stage 1}, the exhaustive sweeping is used, and the initial estimate $\hat{\boldsymbol{\psi}}_{0}$ is obtained in \eqref{eq:initial} by using the orthogonal matching pursuit method (e.g., \cite{Alkhateeb2015Compressed}). This ensures that the initial beam direction $\hat{x}_{0}$ is within the mainlobe set, i.e.,
%\vspace{-2mm}
\begin{equation}\label{eq:mainlobe}
\mathcal{B}\left(x_{0}\right) \overset{\Delta}{=} \Big( x_{0} - \frac{\lambda}{Md}, x_{0} +  \frac{\lambda}{Md}\Big).
%\vspace{-2mm}
\end{equation}

% In \emph{Stage 1}, $M$ pilots are sent successively (see Fig. \ref{fig:frame}). Then, an initial estimate $\hat{\boldsymbol{\psi}}_{0}\!=\![\hat{\beta}^\text{re}_{0}, \hat{\beta}^\text{im}_{0}, \hat{x}_{0}]^\text{T}$ will be obtained. Due to non-convexity of Problem \eqref{eq:problem}, a good initial estimate $\hat{x}_{0}$ is quite important for the success of tracking, which should be within the mainlobe set $\mathcal{B}\left(x_{0}\right) \overset{\Delta}{=} \left( x_{0} - \frac{\lambda}{Md}, x_{0} +  \frac{\lambda}{Md}\right)$. To achieve this goal, the exhaustive sweeping is utilized to thoroughly observe the channel, and then motivated by the orthogonal matching pursuit method (e.g., \cite{Alkhateeb2015Compressed}), we use \eqref{eq:initial} to obtain the initial estimate $\hat{\boldsymbol{\psi}}_{0}$ from these observations. 
In \emph{Stage 2}, the recursive tracker is motivated by the following maximization likelihood problem:
%\vspace{-2mm}
\begin{equation}\label{eq_ML_estimator}
\!\underset{\hat{\boldsymbol{\psi}}_n}{\max}\!\left\{\!\underset{\mathbf{W}_n}{\max}~\!\!\!\sum_{i = 1}^n\mathbb{E}\bigg[\!\log p\!\left( \mathbf{y}_i|\hat{\boldsymbol{\psi}}_n,\!\mathbf{W}_i\!\right)\!\!\bigg|\begin{matrix} \hat{\boldsymbol{\psi}}_n,\!\mathbf{W}_1,\!\ldots,\!\mathbf{W}_i, \\ \!\mathbf{y}_1,\!\ldots,\!\mathbf{y}_{i-1}\end{matrix}\bigg]\!\right\}\!\!,\!\!
%\vspace{-2mm}
\end{equation}
where $\mathbf{W}_n\!=\!\left[ \mathbf{w}_{n,1},\!\mathbf{w}_{n,2}\right]$ is subject to \eqref{eq:bf}. We propose a two-layer nested optimization algorithm to find the solution of \eqref{eq_ML_estimator}. In the \emph{outer layer}, we use the stochastic Newton's method to update the estimate $\hat{\boldsymbol{\psi}}_n$, given by \cite{nevel1973stochastic}
%\begin{align}\label{eq:newton} \!\!\!\!\hat{\boldsymbol{\psi}}_n \!=\!\hat{\boldsymbol{\psi}}_{n\!-\!1}\!+\!a_{n} \mathbf{I}(\hat{\boldsymbol{\psi}}_{n\!-\!1},\!\mathbf{W}_n)^{-1}\!\cdot\! \frac{\partial \log p(\mathbf{y}_n| \hat{\boldsymbol{\psi}}_{n\!-\!1},\!\mathbf{W}_n)}{\partial \hat{\boldsymbol{\psi}}_{n\!-\!1}},\!\!\!\!
%\end{align}
%%\vspace{-2.5mm}
\begin{align}\label{eq:newton}
\!\!\!\!\hat{\boldsymbol{\psi}}_n \!=\!&~ \hat{\boldsymbol{\psi}}_{n\!-\!1}\!-\!a_{n} \mathbb{E}\!\left[\mathbf{H}(\hat{\boldsymbol{\psi}}_{n\!-\!1},\!\mathbf{W}_n)\right]^{-1}\!\cdot\!\frac{\partial \log p(\mathbf{y}_n| \hat{\boldsymbol{\psi}}_{n\!-\!1},\!\mathbf{W}_n)}{\partial \hat{\boldsymbol{\psi}}_{n\!-\!1}}\!\!\!\! \nonumber \\ 
=\!&~\hat{\boldsymbol{\psi}}_{n\!-\!1}\!+\!a_{n} \mathbf{I}(\hat{\boldsymbol{\psi}}_{n\!-\!1},\!\mathbf{W}_n)^{-1}\!\cdot\! \frac{\partial \log p(\mathbf{y}_n| \hat{\boldsymbol{\psi}}_{n\!-\!1},\!\mathbf{W}_n)}{\partial \hat{\boldsymbol{\psi}}_{n\!-\!1}},\!\!\!\!
\end{align}
%\vspace{-2mm}

\noindent where $\mathbf{H}(\hat{\boldsymbol{\psi}}_{n\!-\!1},\!\mathbf{W}_n)\!=\!\frac{\partial^2 \log p(\mathbf{y}_n| \hat{\boldsymbol{\psi}}_{n\!-\!1},\!\mathbf{W}_n)}{\partial \hat{\boldsymbol{\psi}}_{n\!-\!1}\partial \hat{\boldsymbol{\psi}}_{n\!-\!1}^\text{T}}$ is the Hessian matrix, $\mathbf{I}(\hat{\boldsymbol{\psi}}_{n\!-\!1},\!\mathbf{W}_n)$ can be calculated by using \eqref{eq:FIM}, $a_n$ is the step-size that will be specified later, and 
%\vspace{-2.5mm}
\begin{equation}\label{eq:gra}
\begin{aligned}
\!\!\!\!\!\!\frac{\partial \log p(\mathbf{y}_n| \hat{\boldsymbol{\psi}}_{n\!-\!1},\!\mathbf{W}_n)}{\partial \hat{\boldsymbol{\psi}}_{n\!-\!1}}\!=\!-\frac{2}{\sigma_0^2}\!\!\left[\begin{matrix} \operatorname{Re}\{{s}^\text{H}\hat{\mathbf{g}}_{n}^\text{H}(\mathbf{y_n}\!-\!{s}\hat{\beta}_{n\!-\!1}\hat{\mathbf{g}}_{n})\} \\
\operatorname{Im}\{{s}^\text{H}\hat{\mathbf{g}}_{n}^\text{H}(\mathbf{y_n}\!-\!{s}\hat{\beta}_{n\!-\!1}\hat{\mathbf{g}}_{n})\} \\
\operatorname{Re}\{{s}^\text{H}\hat{\mathbf{e}}_{n}^\text{H}(\mathbf{y_n}\!-\!{s}\hat{\beta}_{n\!-\!1}\hat{\mathbf{g}}_{n})\} \end{matrix}\right]\!\!,
\end{aligned}
%\vspace{-1mm}
\end{equation}
with $\hat{\mathbf{g}}_{n}\!=\!\mathbf{W}_n^\text{H}\mathbf{a}(\hat{x}_{n\!-\!1})$ and $\hat{\mathbf{e}}_{n}\!=\!\hat{\beta}_{n\!-\!1}\mathbf{W}_n^\text{H}\dot{\mathbf{a}}(\hat{x}_{n\!-\!1})$. Plugging $\mathbf{I}(\hat{\boldsymbol{\psi}}_{n\!-\!1},\!\mathbf{W}_n)$ and \eqref{eq:gra} in \eqref{eq:newton}, we get \eqref{eq:newton2}. In the \emph{inner layer}, it is equivalent to minimize the CRLB to update $\mathbf{W}_n$, i.e.,
%\vspace{-2mm}
\begin{equation}\label{eq:minCRLB2}
\begin{aligned}
\underset{\mathbf{W}_n}{\min}&~\left[ \mathbf{I}(\hat{\boldsymbol{\psi}}_{n-1}, \mathbf{W}_n)^{-1} \right]_{3,3},
\end{aligned}
%\vspace{-2mm}
\end{equation}
which results in \eqref{eq:est_ctrl}.
% where $\mathbf{W}_n\!=\!\left[ \mathbf{w}_{n,1},\!\mathbf{w}_{n,2}\right]$ is subject to \eqref{eq:bf}. Similar to \eqref{eq:control}, the solution of \eqref{eq:minCRLB} is given by
\begin{remark}
{\blue Different from the beam tracking algorithm in \cite{Li2017conf, Li2017super}, the RBCT algorithm uses two pilots and jointly updates the beam direction and channel coefficient in each time-slot.}
\end{remark}

\section{Asymptotic Optimality Analysis}
\label{sec:analysis}

{\red

\begin{figure*}[!t]
\normalsize
\setcounter{mytempeqncnt}{\value{equation}}
\setcounter{equation}{19}
%\vspace{-7mm}
\begin{equation}\label{eq:function_f}\begin{aligned}
\mathbf{f}\left(\hat{\boldsymbol{\psi}}_{n-1}, \boldsymbol{\psi}_n\right) \!\overset{\Delta}{=}\! \mathbb{E}\!\!\left[\!\left.\mathbf{I}(\hat{\boldsymbol{\psi}}_{n\!-\!1},\!\mathbf{W}_n)^{-1}\!\cdot\!\frac{\partial \log p(\mathbf{y}_n| \hat{\boldsymbol{\psi}}_{n-1},\!\mathbf{W}_n)}{\partial \hat{\boldsymbol{\psi}}_{n-1}}\right| \boldsymbol{\psi}_n\!\right] \!=\! -\frac{2}{\sigma_0^2}\mathbf{I}(\hat{\boldsymbol{\psi}}_{n\!-\!1},\!\mathbf{W}_n)^{-1}\!\cdot\!\!\left[\begin{matrix} \operatorname{Re}\{\hat{\mathbf{g}}_{n}^\text{H}(\beta_n\mathbf{W}_n^\text{H}\mathbf{a}(x_n)\!-\!\hat{\beta}_{n\!-\!1}\hat{\mathbf{g}}_{n})\} \\
\operatorname{Im}\{\hat{\mathbf{g}}_{n}^\text{H}(\beta_n\mathbf{W}_n^\text{H}\mathbf{a}(x_n)\!-\!\hat{\beta}_{n\!-\!1}\hat{\mathbf{g}}_{n})\} \\
\operatorname{Re}\{\hat{\mathbf{e}}_{n}^\text{H}(\beta_n\mathbf{W}_n^\text{H}\mathbf{a}(x_n)\!-\!\hat{\beta}_{n\!-\!1}\hat{\mathbf{g}}_{n})\} \end{matrix}\right]\!\!. 
\end{aligned}%\vspace{-2mm}
\end{equation}
\begin{equation}\label{eq:noise_z}\begin{aligned}
\hat{\mathbf{z}}_n &~\overset{\Delta}{=} \mathbf{I}(\hat{\boldsymbol{\psi}}_{n\!-\!1},\!\mathbf{W}_n)^{-1}\!\cdot\!\frac{\partial \log p(\mathbf{y}_n| \hat{\boldsymbol{\psi}}_{n-1},\!\mathbf{W}_n)}{\partial \hat{\boldsymbol{\psi}}_{n-1}} - \mathbf{f}\left(\hat{\boldsymbol{\psi}}_{n-1}, \boldsymbol{\psi}_n\right) 
= -\frac{2|{s}|^2}{\sigma_0^2}\mathbf{I}(\hat{\boldsymbol{\psi}}_{n\!-\!1},\!\mathbf{W}_n)^{-1}\!\cdot\!\left[\begin{matrix} \operatorname{Re}\{{s}^\text{H}\hat{\mathbf{g}}_{n}^\text{H}\mathbf{z}_n\} \\
\operatorname{Im}\{{s}^\text{H}\hat{\mathbf{g}}_{n}^\text{H}\mathbf{z}_n\} \\
\operatorname{Re}\{{s}^\text{H}\hat{\mathbf{e}}_{n}^\text{H}\mathbf{z}_n\} \end{matrix}\right]\!\!.  
\end{aligned}%\vspace{-1mm}
\end{equation}

\setcounter{equation}{23}
\begin{equation}\label{eq:derivative}\begin{aligned}
\!\!\!\!\!\!\frac{\partial \mathbf{f}\left(\hat{\boldsymbol{\psi}}_{n-1}, \boldsymbol{\psi}_n\right)}{\partial \hat{\boldsymbol{\psi}}_{n-1}^\text{T}}\!=\!\frac{\partial \mathbf{I}(\hat{\boldsymbol{\psi}}_{n\!-\!1},\!\mathbf{W}_n)^{-1}}{\partial \hat{\boldsymbol{\psi}}_{n-1}^\text{T}} \!\cdot\! \mathbb{E}\left[\left.\frac{\partial \log p(\mathbf{y}_n| \hat{\boldsymbol{\psi}}_{n-1},\!\mathbf{W}_n)}{\partial \hat{\boldsymbol{\psi}}_{n-1}}\right| \boldsymbol{\psi}_n\right]\!+\!\mathbf{I}(\hat{\boldsymbol{\psi}}_{n\!-\!1},\!\mathbf{W}_n)^{-1} \!\cdot\! \frac{\partial \mathbb{E}\left[\left.\frac{\partial \log p(\mathbf{y}_n| \hat{\boldsymbol{\psi}}_{n-1},\!\mathbf{W}_n)}{\partial \hat{\boldsymbol{\psi}}_{n-1}}\right| \boldsymbol{\psi}_n\right]}{\partial \hat{\boldsymbol{\psi}}_{n-1}^\text{T}}.\!\!\!\!\!
\end{aligned}\end{equation}

\setcounter{equation}{18}

\hrulefill
%\vspace{-2mm}
\end{figure*}

There are multiple stable points for \eqref{eq:newton2}, which correspond to the local optimal estimates for Problem \eqref{eq:problem} \cite{kushner2003stochastic}. Hence Problem \eqref{eq:problem} is non-convex for the estimate $\hat{\boldsymbol{\psi}}_n$. To study these stable points, we rewrite \eqref{eq:newton2} as follows:
%\vspace{-2.5mm}
\begin{equation}\label{eq:newton3}
\hat{\boldsymbol{\psi}}_n = \hat{\boldsymbol{\psi}}_{n\!-\!1} + a_{n} \left(\mathbf{f}\left(\hat{\boldsymbol{\psi}}_{n-1}, \boldsymbol{\psi}_n\right) + \hat{\mathbf{z}}_n \right),
%\vspace{-2.5mm}
\end{equation}
where $\mathbf{f}\left(\hat{\boldsymbol{\psi}}_{n-1}, \boldsymbol{\psi}_n\right)$ is defined in \eqref{eq:function_f},  $\hat{\mathbf{z}}_n$ is defined in \eqref{eq:noise_z}, with $\hat{\mathbf{g}}_{n}\!=\!\mathbf{W}_n^\text{H}\mathbf{a}(\hat{x}_{n\!-\!1})$, $\hat{\mathbf{e}}_{n}\!=\!\hat{\beta}_{n\!-\!1}\mathbf{W}_n^\text{H}\dot{\mathbf{a}}(\hat{x}_{n\!-\!1})$, $l_n\!=\!\|\hat{\mathbf{g}}_{n}\|_2 \|\hat{\mathbf{e}}_{n} \|_2$, $c_n\!=\!\hat{\mathbf{g}}_{n}^\text{H}\hat{\mathbf{e}}_{n}$, and $\mathbf{z}_n\!=\!\left[ z_{n,1}, z_{n,2} \right]^\text{T}$.

A stable point $\hat{\boldsymbol{\psi}}_{n-1}$ should satisfy: 1) $\mathbf{f}\left(\hat{\boldsymbol{\psi}}_{n-1}, \boldsymbol{\psi}_n\right)\!=\!\mathbf{0}$, and 2) $\frac{\partial \mathbf{f}\left(\hat{\boldsymbol{\psi}}_{n-1}, \boldsymbol{\psi}_n\right)}{\partial \hat{\boldsymbol{\psi}}_{n-1}^\text{T}}$ is a negative definite matrix. Let 
%\vspace{-2.5mm}
\begin{equation}
\setcounter{equation}{22}
\!\!\!\!\mathcal{S}_n \!=\! \left\{\! \hat{\boldsymbol{\psi}}_{n-1} \!:\! \mathbf{f}\left(\hat{\boldsymbol{\psi}}_{n-1}, \boldsymbol{\psi}_n\right)\!=\!\mathbf{0}, \frac{\partial \mathbf{f}\left(\hat{\boldsymbol{\psi}}_{n-1}, \boldsymbol{\psi}_n\right)}{\partial \hat{\boldsymbol{\psi}}_{n-1}^\text{T}} \prec 0\!\right\}\!,\!\!\!
%\vspace{-2mm}
\end{equation}
denote the stable points set at time-slot $n$. Then, we can verify $\boldsymbol{\psi}_n\!\in\!\mathcal{S}_n$ as below:
\begin{itemize}
\item[1)] When $\hat{\boldsymbol{\psi}}_{n-1} = \boldsymbol{\psi}_n$, we have $\beta_n\mathbf{W}_n^\text{H}\mathbf{a}(x_n)\!-\!\hat{\beta}_{n\!-\!1}\hat{\mathbf{g}}_{n} = 0.$
Hence, $\mathbf{f}(\boldsymbol{\psi}_n, \boldsymbol{\psi}_n)\!=\!\mathbf{0}$.

\item[2)] From \eqref{eq:function_f}, we can get 
\begin{equation}\label{eq:function_f2}
\begin{aligned}
&~\mathbf{f}\left(\hat{\boldsymbol{\psi}}_{n-1}, \boldsymbol{\psi}_n\right) =  \mathbf{I}(\hat{\boldsymbol{\psi}}_{n\!-\!1},\!\mathbf{W}_n)^{-1} \\
&~~~~~~~~~~~~~\cdot \mathbb{E}\left[\left.\frac{\partial \log p(\mathbf{y}_n| \hat{\boldsymbol{\psi}}_{n-1},\!\mathbf{W}_n)}{\partial \hat{\boldsymbol{\psi}}_{n-1}}\right| \boldsymbol{\psi}_n\right].
\end{aligned}
\end{equation}
Then, the derivative can be obtained by \eqref{eq:derivative}. Similar to 1), the first term in \eqref{eq:derivative} is $\mathbf{0}$ when $\hat{\boldsymbol{\psi}}_{n-1} = \boldsymbol{\psi}_n$. Moreover, the partial derivative $\frac{\partial \mathbb{E}\left[\left.\frac{\partial \log p(\mathbf{y}_n| \hat{\boldsymbol{\psi}}_{n-1},\!\mathbf{W}_n)}{\partial \hat{\boldsymbol{\psi}}_{n-1}}\right| \boldsymbol{\psi}_n\right]}{\partial \hat{\boldsymbol{\psi}}_{n-1}^\text{T}}$ in the second term is equal to $\mathbf{I}(\boldsymbol{\psi}_n,\!\mathbf{W}_n)$ when $\hat{\boldsymbol{\psi}}_{n-1} = \boldsymbol{\psi}_n$. Therefore, when $\hat{\boldsymbol{\psi}}_{n-1} = \boldsymbol{\psi}_n$, we have 
\begin{equation}
\setcounter{equation}{25}
\frac{\partial \mathbf{f}\left(\hat{\boldsymbol{\psi}}_{n-1}, \boldsymbol{\psi}_n\right)}{\partial \hat{\boldsymbol{\psi}}_{n-1}^\text{T}} = - \left[\begin{matrix} 1 & 0 & 0 \\ 0 & 1 & 0 \\ 0 & 0 & 1 \end{matrix} \right] \prec 0.
\end{equation}

\end{itemize}

Note that except for the real direction $x_n$, the antenna array gain is quite low at other local optimal stable points in $\mathcal{S}_n$. Hence, one key challenge is \emph{how to ensure that the RBCT algorithm converges to the real direction $x_n$, instead of other local optimal stable points in $\mathcal{S}_n$}.

} 

In \emph{static} beam tracking, where $\boldsymbol{\psi}_n\!=\!\boldsymbol{\psi}\!=\![\beta^\text{re},\!\beta^\text{im},\!x]^\text{T}$ and $\mathcal{S}_n \!=\! \mathcal{S}\!\overset{\Delta}{=}\! \Big\{ \hat{\boldsymbol{\psi}}_{n-1} : \mathbf{f}\left(\hat{\boldsymbol{\psi}}_{n-1}, \boldsymbol{\psi}\right)\!=\!\mathbf{0}, \frac{\partial \mathbf{f}\left(\hat{\boldsymbol{\psi}}_{n-1}, \boldsymbol{\psi}\right)}{\partial \hat{\boldsymbol{\psi}}_{n-1}^\text{T}} \prec 0\Big\}$, we adopt the diminishing step-sizes, given by \cite{nevel1973stochastic,kushner2003stochastic, borkar2008stochastic}
%\vspace{-1.5mm}
\begin{equation}\label{eq:stepsize}
a_n = \frac{\alpha}{n + N_0}, ~~n = 1, 2, \ldots,
%\vspace{-1.5mm}
\end{equation}
where $\alpha\!>\!0$ and $N_0\!\ge\!0$. We use the stochastic approximation and recursive estimation
theory \cite{nevel1973stochastic,kushner2003stochastic, borkar2008stochastic} to analyze the RBCT algorithm.{\blue To support the more general joint beam and channel tracking scenario than \cite{Li2017conf, Li2017super}, three new theorems are developed to resolve the challenge mentioned above:}
%%\vspace{-2.5mm}
\begin{theorem}[\textbf{Convergence to Stable Points}]\label{th:convergence}
If $a_n$ is given by (\ref{eq:stepsize}) with any $\alpha > 0$ and $N_0 \ge 0$, then $\hat{\boldsymbol{\psi}}_n$  converges to a unique point within $\mathcal{S}$ with probability one.
%\vspace{-2mm}
\end{theorem}
\begin{proof}
See the detailed proof in Appendix \ref{proof:convergence}.
%\vspace{-2mm}
\end{proof}

Hence, for general step-size parameters $\alpha$ and $N_0$ in \eqref{eq:stepsize}, $\hat{x}_n$ converges to a stable point in $\mathcal{S}$.
%\vspace{-2mm}
\begin{theorem}[\textbf{Convergence to the Real Beam Direction $x$}]\label{th:lock}
If (i) $\hat{x}_{0}\!\in\!\mathcal{B}\left(x\right)$, (ii) $a_n$ is given by (\ref{eq:stepsize}) with any $\alpha\!>\!0$,
then there exist $N_0\!\ge\!0$ and $C\!>\!0$ such that
%\vspace{-4mm}
\begin{equation}\label{eq:lock}
P\left( \left. \hat{x}_n \rightarrow x \right| \hat{x}_{0} \in \mathcal{B}\left(x\right) \right) \geq 1 - 6e^{-\frac{C|s|^2}{\alpha^2\sigma_0^2}}.
%\vspace{-2mm}
\end{equation}
\end{theorem}
\begin{proof}
See the detailed proof in Appendix \ref{proof:lock}.
%\vspace{-2mm}
\end{proof}

By Theorem \ref{th:lock}, if the initial point $\hat{x}_0$ is in the mainlobe $\mathcal{B}(x)$, the probability that $\hat{x}_n$ does not converge to $x$ decades \emph{exponentially} with respect to ${|s|^2}/{\alpha^2\sigma_0^2}$. Hence, one can increase the transmit SNR ${|s|^2}/{\sigma_0^2}$ and reduce the step-size parameter $\alpha$ to ensure $\hat{x}_n\!\rightarrow\!x$ with high probability. 

%\vspace{-2mm}
\begin{theorem}[\textbf{Convergence to $x$ with the Minimum MSE}]
\label{th:normal}
If (i) $a_n$ is given by (\ref{eq:stepsize}) with $\alpha = 1$ and any $N_0 \ge 0$, and (ii) $\hat{\boldsymbol{\psi}}_n \rightarrow \boldsymbol{\psi}$,
then
%\begin{equation}\label{eq:normal}
%	\sqrt{n}\left(\hat{x}_n - x\right) \overset{d}{\rightarrow} \mathcal{N} \left(0, I_{\max}^{-1}\right),
%\end{equation}
%as $n \rightarrow \infty$, where $\overset{d}{\rightarrow}$ represents convergence in conditional distribution, and $I_{\max}$ is defined in \eqref{eq:MMSE}. In addition,
%\vspace{-2mm}
\begin{equation}\label{eq:normal2}
	\lim_{n\rightarrow\infty}~\!n~\!\mathbb{E}\left[\left(\hat{x}_n - x\right)^2\big| \hat{\boldsymbol{\psi}}_n \rightarrow \boldsymbol{\psi}\right] = \left[ \mathbf{I}(\boldsymbol{\psi}, \mathbf{W}^*)^{-1} \right]_{3,3}.
	%\vspace{-2mm}
\end{equation}
\end{theorem}
\begin{proof}
See the detailed proof in Appendix \ref{proof:normal}.
%\vspace{-2mm}
\end{proof}

%\begin{proof}[Proof Description for Theorem 1-3]
% Similar to \cite{Li2017conf, Li2017super}, we use Theorem 5.2.1 of \cite{kushner2003stochastic}, Chapter 4 of \cite{borkar2008stochastic}, and Theorem 6.6.1 of \cite{nevel1973stochastic} to prove Theorem 1-3, respectively. The main difference from \cite{Li2017conf, Li2017super} is that the RBCT algorithm considers vector variable/function, rather than the scalar ones in \cite{Li2017conf, Li2017super}. Hence, the main proof structures are almost the same, while the parts that involve these vectors are different. The proof is omitted due to space limitation.
%\end{proof}
%%\vspace{-2mm}

Theorem \ref{th:normal} tells us that $\alpha$ should not be too small: If $\alpha=1$, then the minimum CRLB on the RHS of \eqref{eq:CRLB} is achieved asymptotically with high probability. 

% By Theorem 1-3, we can observe that the RBCT algorithm can ensure that the estimate $\hat{\boldsymbol{\psi}}_n$ converges to the real parameters $\boldsymbol{\psi}$ and the beam direction MSE achieves the minimum CRLB asymptotically, with high probability. 

\section{Numerical Results}
\label{sec:numerical}

\begin{figure}[t]
\centering
\includegraphics[width=8.6cm]{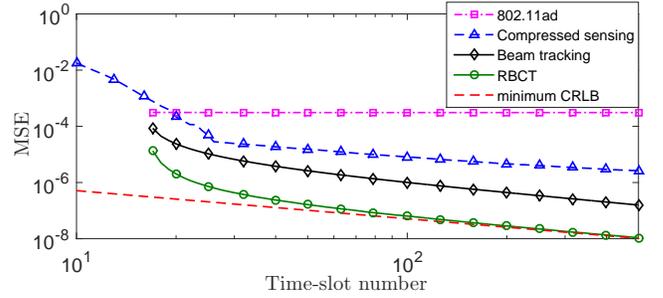}
%\vspace{-3mm}
\caption{MSE vs. time-slot number in static scenarios.}
%\vspace{-2mm}
\label{fig:static_tracking}
\end{figure}

We compare the RBCT algorithm with three reference algorithms: the compressed sensing algorithm \cite{Alkhateeb2015Compressed}, the IEEE 802.11ad algorithm \cite{IEEE80211ad}, and the beam tracking algorithm \cite{Li2017super}. The first two algorithms have the same configuration as that in Section VI of \cite{Li2017super}. The third one uses the same training beamforming vectors as the RBCT algorithm, i.e., in each time-slot, it receives two pilots with the beamforming vectors in \eqref{eq:est_ctrl}, and the beam direction is tracked by using both observations. Moreover, its channel coefficient is obtained with a least square estimator by using these observations. Consider the system model in Section \ref{sec:model} with $M \!=\! 32$ antennas, and the antenna spacing is $d \!=\! 0.5\lambda$. The pilot symbol is ${s} \!=\! \frac{1+j}{2}$, and the transmit SNR $\frac{|{s}|^2}{\sigma^2}$ is set as $5\text{dB}$. To ensure fairness, we assume that 2 pilot symbols are received in each time-slot, hence all the algorithms have the same pilot overhead.

In static scenarios, we set the step-size as $a_n\!=\!\frac{1}{n}, n\!\ge\!1$. The real AoA $\theta$ is randomly generated by a uniform distribution on $[-90^\circ, 90^\circ]$ in each realization, and the results are averaged over 10000 random realizations. Figure \ref{fig:static_tracking} plots the MSE over time. It can be observed that the MSE of the RBCT algorithm converges to the minimum CRLB in \eqref{eq:CRLB}, which is much smaller than the reference algorithms.

\begin{figure}[t]
\centering
\includegraphics[width=8.6cm]{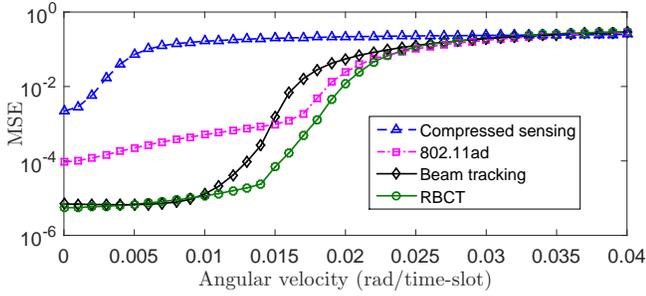}
%\vspace{-4mm}
\caption{MSE vs. angular velocity in dynamic scenarios.}
%\vspace{-2mm}
\label{fig:dynamic_tracking_MSE}
\end{figure}

\begin{figure}[t]
\centering
\includegraphics[width=8.6cm]{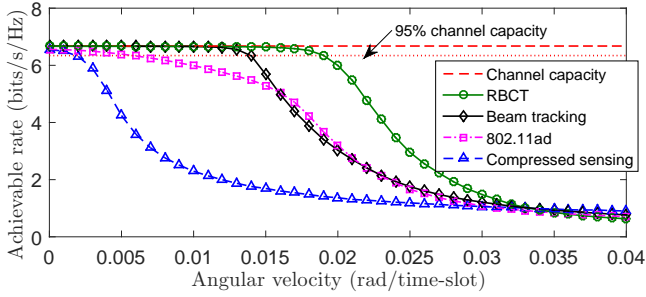}
%\vspace{-4mm}
\caption{Data rate vs. angular velocity in dynamic scenarios.}
%\vspace{-4mm}
\label{fig:dynamic_tracking_rate}
\end{figure}

In dynamic scenarios, we set the step-size as a constant value, i.e., $a_n\!=\!1, n\!\ge\!1$. The channel variation is modeled as: The AoA $\theta_n\!=\!\theta_{n-1}\!+\!\delta_{n-1}\!\cdot\!\omega$ where $\theta_0=0$, $\delta_n\in\{-1,\!1\}$ denotes the rotation direction, and $\omega\!\in\![0, 0.04]$ is a fixed angular velocity. The rotation direction $\delta_n$ is chosen such that $\theta_n$ varies within $[-\frac{\pi}{3},\!\frac{\pi}{3}]$. The channel coefficient $\beta_n (\mathbb{E}\left[|\beta_n|^2\right]\!=\!1)$ is subject to Rician fading with a K-factor $\kappa\!=\!15\text{dB}$, according to the channel model proposed in \cite{Samimi2016Kfactor}. In Fig. \ref{fig:dynamic_tracking_MSE} and \ref{fig:dynamic_tracking_rate}, one can observe that the RBCT algorithm can support much higher angular velocities and data rates than other algorithms. According to Fig. \ref{fig:dynamic_tracking_rate}, the RBCT algorithm can achieve 95\% of channel capacity when the angular velocity is 0.19rad (1.09 degrees) per time-slot. If 5 time-slots last for one second, i.e., 10 pilots per second received, then the RBCT algorithm is capable of tracking a mobile moving at an angular velocity of 5.45 degrees per second and achieving over 95\% of channel capacity.
%\vspace{-2.5mm}
\section{Conclusion}
%\vspace{-0.5mm}
\label{sec:conclusion}
We have developed a joint beam and channel tracking algorithm {\blue for mmWave phased antenna arrays}, and established its convergence and asymptomatic optimality. Our simulation results show that the proposed algorithm can achieve much faster tracking speed, lower beam tracking error, and higher data rate than several state-of-the-art algorithms, with the same pilot overhead.

\bibliographystyle{IEEEtran}
\bibliography{reference}

\appendices

\section{Proof of Theorem~\ref{th:convergence}}\label{proof:convergence}

Before providing the proof, let us provide some useful definitions. In \emph{static} beam tracking, where $\boldsymbol{\psi}_n\!=\!\boldsymbol{\psi}\!=\![\beta^\text{re},\!\beta^\text{im},\!x]^\text{T}$ and $\mathcal{S}_n \!=\! \mathcal{S}\!\overset{\Delta}{=}\! \Big\{ \hat{\boldsymbol{\psi}}_{n-1} \!:\! \mathbf{f}\left(\hat{\boldsymbol{\psi}}_{n-1}, \boldsymbol{\psi}\right)\!=\!\mathbf{0}, \frac{\partial \mathbf{f}\left(\hat{\boldsymbol{\psi}}_{n-1}, \boldsymbol{\psi}\right)}{\partial \hat{\boldsymbol{\psi}}_{n-1}^\text{T}} \prec 0\Big\}$, recall the recursive procedure \eqref{eq:newton3}:
\begin{align}\label{eq:newton4}
\hat{\boldsymbol{\psi}}_n = \hat{\boldsymbol{\psi}}_{n\!-\!1} + a_{n} \left(\mathbf{f}\left(\hat{\boldsymbol{\psi}}_{n-1}, \boldsymbol{\psi}\right) + \hat{\mathbf{z}}_n \right) ,
\end{align}
where $\mathbf{f}\left(\hat{\boldsymbol{\psi}}_{n-1}, \boldsymbol{\psi}\right)$ and $\hat{\mathbf{z}}_n$ are given in \eqref{eq:function_f} and \eqref{eq:noise_z} separately. From \eqref{eq:noise_z}, we have
\begin{align}\label{eq:noise_z2}\hat{\mathbf{z}}_n \sim \mathcal{N}\left( \mathbf{0}, \mathbf{I}(\hat{\boldsymbol{\psi}}_{n\!-\!1},\!\mathbf{W}_n)^{-1} \right),\end{align}
where $\mathbb{E}\left[ \hat{\mathbf{z}}_n \right] = \mathbf{0}$, and $\mathbf{I}(\hat{\boldsymbol{\psi}}_{n\!-\!1},\!\mathbf{W}_n)^{-1}$ is the covariance matrix of $\hat{\mathbf{z}}_n$ calculated by \eqref{eq:covariance}. In \eqref{eq:covariance}, the step $(a)$ can be obtain as follows:

\begin{figure*}[!t]
\normalsize
\setcounter{mytempeqncnt}{\value{equation}}
% \setcounter{equation}{14}
%%\vspace{-5mm}
\begin{align}\label{eq:covariance}
% \frac{1}{\left\|\hat{\mathbf{g}}_{n}\right\|^4_2 (l_n^2 \!-\! |c_n|^2)^2}\!\cdot\!\left[\!\begin{matrix} l_n^2\!-\!\operatorname{Im}\{c_n\}^2  \!&\! \operatorname{Re}\{c_n\}\operatorname{Im}\{c_n\} \!&\! -\|\hat{\mathbf{g}}_{n}\|^2_2\operatorname{Re}\{c_n\} \\ 
% \operatorname{Re}\{c_n\}\operatorname{Im}\{c_n\} \!&\! l_n^2\!-\!\operatorname{Re}\{c_n\}^2  \!&\! -\|\hat{\mathbf{g}}_{n}\|^2_2\operatorname{Im}\{c_n\} \\
%-\|\hat{\mathbf{g}}_{n}\|^2_2\operatorname{Re}\{c_n\} \!&\! -\|\hat{\mathbf{g}}_{n}\|^2_2\operatorname{Im}\{c_n\} \!&\! \|\hat{\mathbf{g}}_{n}\|^4_2 \end{matrix}\!\right]
\mathbb{E}\left[ \left(\hat{\mathbf{z}}_n - \mathbb{E}\left[ \hat{\mathbf{z}}_n \right] \right) \left(\hat{\mathbf{z}}_n - \mathbb{E}\left[ \hat{\mathbf{z}}_n \right] \right)^\text{T}\right] = &~\frac{4}{\sigma_0^4}\mathbf{I}(\hat{\boldsymbol{\psi}}_{n\!-\!1},\!\mathbf{W}_n)^{-1}\!\cdot\!\mathbb{E}\!\left\{\!\!\left[\begin{matrix} \operatorname{Re}\{{s}^\text{H}\hat{\mathbf{g}}_{n}^\text{H}\mathbf{z}_n\} \\
\operatorname{Im}\{{s}^\text{H}\hat{\mathbf{g}}_{n}^\text{H}\mathbf{z}_n\} \\
\operatorname{Re}\{{s}^\text{H}\hat{\mathbf{e}}_{n}^\text{H}\mathbf{z}_n\} \end{matrix}\right]\!\!\cdot\!\!\left[\begin{matrix} \operatorname{Re}\{{s}^\text{H}\hat{\mathbf{g}}_{n}^\text{H}\mathbf{z}_n\} \\
\operatorname{Im}\{{s}^\text{H}\hat{\mathbf{g}}_{n}^\text{H}\mathbf{z}_n\} \\
\operatorname{Re}\{{s}^\text{H}\hat{\mathbf{e}}_{n}^\text{H}\mathbf{z}_n\} \end{matrix}\right]^\text{\!\!T} \!\right\}
\!\cdot\!\mathbf{I}(\hat{\boldsymbol{\psi}}_{n\!-\!1},\!\mathbf{W}_n)^{-1} \\
\overset{(a)}{=} &~\mathbf{I}(\hat{\boldsymbol{\psi}}_{n\!-\!1},\!\mathbf{W}_n)^{-1}. \nonumber
%\left[\!\begin{matrix} l_n^2\!-\!\operatorname{Im}\{c_n\}^2  \!&\! \operatorname{Re}\{c_n\}\operatorname{Im}\{c_n\} \!&\! -\|\hat{\mathbf{g}}_{n}\|^2_2\operatorname{Re}\{c_n\} \\ 
%\operatorname{Re}\{c_n\}\operatorname{Im}\{c_n\} \!&\! l_n^2\!-\!\operatorname{Re}\{c_n\}^2  \!&\! -\|\hat{\mathbf{g}}_{n}\|^2_2\operatorname{Im}\{c_n\} \\
%-\|\hat{\mathbf{g}}_{n}\|^2_2\operatorname{Re}\{c_n\} \!&\! -\|\hat{\mathbf{g}}_{n}\|^2_2\operatorname{Im}\{c_n\} \!&\! \|\hat{\mathbf{g}}_{n}\|^4_2 \end{matrix}\!\right]. \nonumber
\end{align}

\hrulefill
%%\vspace{-4mm}
\end{figure*}

\begin{itemize}
\item Since $\mathbf{z}_n\!=\!\left[ z_{n,1}, z_{n,2} \right]^\text{T}$ consists of two \emph{i.i.d.} circularly symmetric complex Gaussian random variables, we  get 
\begin{align}\label{eq:variance_1}
{s}^\text{H}\hat{\mathbf{g}}_{n}^\text{H}\mathbf{z}_n \sim \mathcal{CN}\left( 0, \left\|{s}\hat{\mathbf{g}}_{n}\right\|_2^2 \sigma_0^2 \right),
\end{align}
and
\begin{align}\label{eq:variance_2}
{s}^\text{H}\hat{\mathbf{e}}_{n}^\text{H}\mathbf{z}_n \sim \mathcal{CN}\left( 0, \left\|{s}\hat{\mathbf{e}}_{n}\right\|_2^2 \sigma_0^2 \right).
\end{align}

\item By splitting the real part and imaginary part, we obtain
\begin{equation}\label{eq:real_imaginary}
\!\!\!\!\!\!\!\!\!\left\{
\begin{aligned}
&\operatorname{Re}\{{s}^\text{H}\hat{\mathbf{g}}_{n}^\text{H}\mathbf{z}_n\} \!=\! \operatorname{Re}\{{s}^\text{H}\hat{\mathbf{g}}_{n}^\text{H}\}\operatorname{Re}\{\mathbf{z}_n\} \!-\! \operatorname{Im}\{{s}^\text{H}\hat{\mathbf{g}}_{n}^\text{H}\}\operatorname{Im}\{\mathbf{z}_n\}, \\
&\operatorname{Im}\{{s}^\text{H}\hat{\mathbf{g}}_{n}^\text{H}\mathbf{z}_n\} \!=\! \operatorname{Re}\{{s}^\text{H}\hat{\mathbf{g}}_{n}^\text{H}\}\operatorname{Im}\{\mathbf{z}_n\} \!+\! \operatorname{Im}\{{s}^\text{H}\hat{\mathbf{g}}_{n}^\text{H}\}\operatorname{Re}\{\mathbf{z}_n\}, \\
&\operatorname{Re}\{{s}^\text{H}\hat{\mathbf{e}}_{n}^\text{H}\mathbf{z}_n\} \!=\! \operatorname{Re}\{{s}^\text{H}\hat{\mathbf{e}}_{n}^\text{H}\}\operatorname{Re}\{\mathbf{z}_n\} \!-\! \operatorname{Im}\{{s}^\text{H}\hat{\mathbf{e}}_{n}^\text{H}\}\operatorname{Im}\{\mathbf{z}_n\}, \\
&\operatorname{Re}\{{s}^\text{H}\hat{\mathbf{g}}_{n}^\text{H}{s}\hat{\mathbf{e}}_{n}\} \!=\! |{s}|^2\operatorname{Re}\{\hat{\mathbf{g}}_{n}^\text{H}\hat{\mathbf{e}}_{n}\} \\
&~~~~~~~~~~\!= \operatorname{Re}\{{s}^\text{H}\hat{\mathbf{g}}_{n}^\text{H}\}\operatorname{Re}\{{s}\hat{\mathbf{e}}_{n}\} \!-\! \operatorname{Im}\{{s}^\text{H}\hat{\mathbf{g}}_{n}^\text{H}\}\operatorname{Im}\{{s}\hat{\mathbf{e}}_{n}\}, \\
&\operatorname{Im}\{{s}^\text{H}\hat{\mathbf{g}}_{n}^\text{H}{s}\hat{\mathbf{e}}_{n}\} \!=\! |{s}|^2\operatorname{Im}\{\hat{\mathbf{g}}_{n}^\text{H}\hat{\mathbf{e}}_{n}\} \\
&~~~~~~~~~~\!= \operatorname{Re}\{{s}^\text{H}\hat{\mathbf{g}}_{n}^\text{H}\}\operatorname{Im}\{{s}\hat{\mathbf{e}}_{n}\} \!+\! \operatorname{Im}\{{s}^\text{H}\hat{\mathbf{g}}_{n}^\text{H}\}\operatorname{Re}\{{s}\hat{\mathbf{e}}_{n}\}.
\end{aligned}\right.\!\!\!\!\!
\end{equation}

\item Combining \eqref{eq:variance_1}, \eqref{eq:variance_2} and \eqref{eq:real_imaginary} yields
\begin{equation}
\!\!\!\!\!\!\left\{
\begin{aligned}
&\mathbb{E}\left[\operatorname{Re}\{{s}^\text{H}\hat{\mathbf{g}}_{n}^\text{H}\mathbf{z}_n\}^2\right] \!=\! \mathbb{E}\left[\operatorname{Im}\{{s}^\text{H}\hat{\mathbf{g}}_{n}^\text{H}\mathbf{z}_n\}^2\right] \!=\! \frac{|{s}|^2 \sigma_0^2}{2}\left\|\hat{\mathbf{g}}_{n}\right\|_2^2, \\
&\mathbb{E}\left[\operatorname{Re}\{{s}^\text{H}\hat{\mathbf{e}}_{n}^\text{H}\mathbf{z}_n\}^2\right] \!=\! \frac{|{s}|^2 \sigma_0^2}{2}\left\|\hat{\mathbf{e}}_{n}\right\|_2^2, \\
&\mathbb{E}\left[\operatorname{Re}\{{s}^\text{H}\hat{\mathbf{g}}_{n}^\text{H}\mathbf{z}_n\}\!\cdot\!\operatorname{Im}\{{s}^\text{H}\hat{\mathbf{g}}_{n}^\text{H}\mathbf{z}_n\}\right] = 0, \\
&\mathbb{E}\left[\operatorname{Re}\{{s}^\text{H}\hat{\mathbf{g}}_{n}^\text{H}\mathbf{z}_n\}\!\cdot\!\operatorname{Re}\{{s}^\text{H}\hat{\mathbf{e}}_{n}^\text{H}\mathbf{z}_n\} \right] = \frac{|{s}|^2 \sigma_0^2}{2}\operatorname{Re}\{\hat{\mathbf{g}}_{n}^\text{H}\hat{\mathbf{e}}_{n}\} , \\
& \mathbb{E}\left[\operatorname{Im}\{{s}^\text{H}\hat{\mathbf{g}}_{n}^\text{H}\mathbf{z}_n\}\!\cdot\!\operatorname{Re}\{{s}^\text{H}\hat{\mathbf{e}}_{n}^\text{H}\mathbf{z}_n\} \right] = \frac{|{s}|^2 \sigma_0^2}{2}\operatorname{Im}\{\hat{\mathbf{g}}_{n}^\text{H}\hat{\mathbf{e}}_{n}\}.
\end{aligned}\right.\!\!\!\!\!
\end{equation}
Hence, we have
\begin{align}\label{eq:expectation_noise}
\mathbb{E}\!\left\{\!\!\left[\begin{matrix} \operatorname{Re}\{{s}^\text{H}\hat{\mathbf{g}}_{n}^\text{H}\mathbf{z}_n\} \\
\operatorname{Im}\{{s}^\text{H}\hat{\mathbf{g}}_{n}^\text{H}\mathbf{z}_n\} \\
\operatorname{Re}\{{s}^\text{H}\hat{\mathbf{e}}_{n}^\text{H}\mathbf{z}_n\} \end{matrix}\right]\!\!\cdot\!\!\left[\begin{matrix} \operatorname{Re}\{{s}^\text{H}\hat{\mathbf{g}}_{n}^\text{H}\mathbf{z}_n\} \\
\operatorname{Im}\{{s}^\text{H}\hat{\mathbf{g}}_{n}^\text{H}\mathbf{z}_n\} \\
\operatorname{Re}\{{s}^\text{H}\hat{\mathbf{e}}_{n}^\text{H}\mathbf{z}_n\} \end{matrix}\right]^\text{\!\!T} \!\right\} \!=\!\frac{\sigma_0^4}{4} \mathbf{I}(\hat{\boldsymbol{\psi}}_{n\!-\!1},\!\mathbf{W}_n).
\end{align}

\item Plugging \eqref{eq:expectation_noise} into \eqref{eq:covariance} yields the result of step $(a)$.
\end{itemize}

Let $\{ \mathcal{G}_n: n \ge 0 \}$ be an increasing sequence of $\sigma$-fields of $\{ \hat{\boldsymbol{\psi}}_{0}, \hat{\boldsymbol{\psi}}_{1}, \hat{\boldsymbol{\psi}}_{2}, \ldots \}$, i.e., $\mathcal{G}_{n-1}\!\subset\!\mathcal{G}_n$, where $\mathcal{G}_{0} \!\overset{\Delta}{=}\! \sigma(\hat{\boldsymbol{\psi}}_{0})$ and $\mathcal{G}_n \!\overset{\Delta}{=}\! \sigma(\hat{\boldsymbol{\psi}}_{0}, \hat{\mathbf{z}}_{1}, \ldots, \hat{\mathbf{z}}_{n}) $ for $n \ge 1$. Because the $\hat{\mathbf{z}}_n$'s are composed of \emph{i.i.d.} circularly symmetric complex Gaussian random variables with zero mean, $\hat{\mathbf{z}}_n$ is independent of $\mathcal{G}_{n-1}$, and $\hat{\boldsymbol{\psi}}_{n-1} \!\in\! \mathcal{G}_{n-1}$. Hence, we have
\begin{align}\label{eq_fil2}
	&~\mathbb{E} \left[ \left. \mathbf{f}\left(\hat{\boldsymbol{\psi}}_{n-1}, \boldsymbol{\psi}\right) + \hat{\mathbf{z}}_n \right| \mathcal{G}_{n-1} \right] \\
	= &~\mathbb{E} \left[ \left. \mathbf{f}\left(\hat{\boldsymbol{\psi}}_{n-1}, \boldsymbol{\psi}\right) \right| \mathcal{G}_{n-1} \right] + \mathbb{E} \left[ \left. \hat{\mathbf{z}}_n \right| \mathcal{G}_{n-1} \right] = \mathbf{f}\left(\hat{\boldsymbol{\psi}}_{n-1}, \boldsymbol{\psi}\right), \nonumber
\end{align}
for $n \ge 1$.

Theorem 5.2.1 in \cite[Section 5.2.1]{kushner2003stochastic} provided the sufficient conditions under which $\hat{x}_n$ converges to a unique point within a set of stable points with probability one. We will prove that when the step-size $a_n$ is given by (\ref{eq:stepsize}) with any $\alpha > 0$ and $N_0 \ge 0$, our algorithm satisfies its sufficient conditions below:
\begin{itemize}
\item[1)] Step-size requirements:
\begin{equation}\left\{\begin{aligned}&a_n = \frac{\alpha}{n + N_0} \rightarrow 0, \\
& \sum\limits_{n=1}^\infty a_n = \sum\limits_{n=1}^\infty \frac{\alpha}{n+N_0} = \infty, \\
& \sum\limits_{n=1}^\infty a_n^2 = \sum\limits_{n=1}^\infty \frac{\alpha^2}{(n+N_0)^2} \le \sum\limits_{i=1}^\infty \frac{\alpha^2}{i^2} < \infty.\end{aligned}\right.\end{equation}

\item[2)] We need to prove that $\sup\nolimits_n \mathbb{E} \left[ \left\|\mathbf{f}\left(\hat{\boldsymbol{\psi}}_{n-1}, \boldsymbol{\psi}\right) + \hat{\mathbf{z}}_n\right\|_2^2 \right] < \infty$. \\
From \eqref{eq:newton4} and \eqref{eq:noise_z2}, we have
\begin{align}\label{eq_expectation_yn} & \mathbb{E} \left[ \left\|\mathbf{f}\left(\hat{\boldsymbol{\psi}}_{n-1}, \boldsymbol{\psi}\right) + \hat{\mathbf{z}}_n\right\|_2^2 \right] \\
= & \mathbb{E} \left[ \left\|\mathbf{f}\left(\hat{\boldsymbol{\psi}}_{n-1}, \boldsymbol{\psi}\right)\right\|_2^2 + 2 \mathbf{f}\left(\hat{\boldsymbol{\psi}}_{n-1}, \boldsymbol{\psi}\right)^\text{T} \hat{\mathbf{z}}_n + \left\|\hat{\mathbf{z}}_n\right\|_2^2 \right] \nonumber \\ \overset{(a)}{=} & \mathbb{E} \left[ \left\|\mathbf{f}\left(\hat{\boldsymbol{\psi}}_{n-1}, \boldsymbol{\psi}\right)\right\|_2^2 \right] + \operatorname{tr}\left(\mathbf{I}(\hat{\boldsymbol{\psi}}_{n\!-\!1},\!\mathbf{W}_n)^{-1}\right),\nonumber \end{align}
where step $(a)$ is due to \eqref{eq:noise_z2} and that $\hat{\mathbf{z}}_n$ is independent of $\mathbf{f}\left(\hat{\boldsymbol{\psi}}_{n-1}, \boldsymbol{\psi}\right)$. \\ From \eqref{eq:gra} and (\ref{eq:function_f2}), we have
\begin{align}\label{eq_ub_fx}\left\|\mathbf{f}\left(\hat{\boldsymbol{\psi}}_{n-1}, \boldsymbol{\psi}\right)\right\|_2^2 \le &~\left\|\mathbf{I}(\hat{\boldsymbol{\psi}}_{n\!-\!1},\!\mathbf{W}_n)^{-1}\right\|_\text{F}^2 \\
&~\!\!\!\!\!\!\!\!\!\!\!\!\!\!\!\!\!\!\!\!\!\!\!\!\!\!\!\!\!\!\!\!\!\cdot\!\left\|\frac{2|{s}|^2}{\sigma_0^2}\!\!\left[\begin{matrix} \operatorname{Re}\{\hat{\mathbf{g}}_{n}^\text{H}(\beta_n\mathbf{W}_n^\text{H}\mathbf{a}(x_n)\!-\!\hat{\beta}_{n\!-\!1}\hat{\mathbf{g}}_{n})\} \\
\operatorname{Im}\{\hat{\mathbf{g}}_{n}^\text{H}(\beta_n\mathbf{W}_n^\text{H}\mathbf{a}(x_n)\!-\!\hat{\beta}_{n\!-\!1}\hat{\mathbf{g}}_{n})\} \\
\operatorname{Re}\{\hat{\mathbf{e}}_{n}^\text{H}(\beta_n\mathbf{W}_n^\text{H}\mathbf{a}(x_n)\!-\!\hat{\beta}_{n\!-\!1}\hat{\mathbf{g}}_{n})\} \end{matrix}\right]\right\|_2^2. \nonumber \end{align}
Due to that the Fisher information matrix is invertible, we get 
\begin{align}\label{eq:proof1_1}\left\|\mathbf{I}(\hat{\boldsymbol{\psi}}_{n\!-\!1},\!\mathbf{W}_n)^{-1}\right\|_\text{F}^2 < \infty.\end{align} 
In addition, since $\mathbf{W}_{n}\!=\!\left[\mathbf{w}_{n,1},\mathbf{w}_{n,2}\right]$, $\hat{\mathbf{g}}_{n}\!=\!\mathbf{W}_n^\text{H}\mathbf{a}(\hat{x}_{n\!-\!1})$, $\hat{\mathbf{e}}_{n}\!=\!\hat{\beta}_{n\!-\!1}\mathbf{W}_n^\text{H}\dot{\mathbf{a}}(\hat{x}_{n\!-\!1})$,
\begin{align}
\left|\mathbf{w}_{n,i}^\text{H} \mathbf{a}(x)\right| &~= \left|\sum_{m=1}^M \frac{1}{\sqrt{M}}e^{-j (\frac{2\pi d}{\lambda} x - w_{mn,i})}\right| \\
&~\le \sum_{m=1}^M \frac{1}{\sqrt{M}} \left|e^{-j (\frac{2\pi d}{\lambda} (m-1)x - w_{mn,i})}\right| \nonumber \\
&~= \sqrt{M}  < \infty, \nonumber
\end{align}
and
\begin{align}
\left|\mathbf{w}_{n,i}^\text{H} \dot{\mathbf{a}}(x)\right| &~= \left|\sum_{m=1}^M -j \frac{2\pi d (m-1)}{\lambda\sqrt{M}} e^{-j (\frac{2\pi d}{\lambda} x - w_{mn,i})}\right| \nonumber \\
&~\le \sum_{m=1}^M \frac{2\pi d (m-1)}{\lambda\sqrt{M}} \left|e^{-j (\frac{2\pi d}{\lambda} (m-1)x - w_{mn,i})}\right| \nonumber \\
&~= \frac{\pi d \sqrt{M}(M-1)}{\lambda}  < \infty, 
\end{align}
for $i = 1, 2$ and all possible $x$, we can get
\begin{align}\label{eq:proof1_2}
\!\!\!\left\|\frac{2|{s}|^2}{\sigma_0^2}\!\!\left[\begin{matrix} \operatorname{Re}\{\hat{\mathbf{g}}_{n}^\text{H}(\beta_n\mathbf{W}_n^\text{H}\mathbf{a}(x_n)\!-\!\hat{\beta}_{n\!-\!1}\hat{\mathbf{g}}_{n})\} \\
\operatorname{Im}\{\hat{\mathbf{g}}_{n}^\text{H}(\beta_n\mathbf{W}_n^\text{H}\mathbf{a}(x_n)\!-\!\hat{\beta}_{n\!-\!1}\hat{\mathbf{g}}_{n})\} \\
\operatorname{Re}\{\hat{\mathbf{e}}_{n}^\text{H}(\beta_n\mathbf{W}_n^\text{H}\mathbf{a}(x_n)\!-\!\hat{\beta}_{n\!-\!1}\hat{\mathbf{g}}_{n})\} \end{matrix}\right]\right\|_2^2\!<\! \infty.\!\!\!
\end{align}
Hence, combining \eqref{eq:proof1_1} and \eqref{eq:proof1_2}, we have
\begin{align}\label{eq:expectation_fx}\mathbb{E} \left[ \left\|\mathbf{f}\left(\hat{\boldsymbol{\psi}}_{n-1}, \boldsymbol{\psi}\right)\right\|_2^2 \right] < \infty. \end{align} 
From \eqref{eq:proof1_1}, we can get $ \operatorname{tr}\left(\mathbf{I}(\hat{\boldsymbol{\psi}}_{n\!-\!1},\!\mathbf{W}_n)^{-1}\right)  < \infty$. Then, we can obtain that
\begin{align}\sup\nolimits_n \mathbb{E} \left[ \left\|\mathbf{f}\left(\hat{\boldsymbol{\psi}}_{n-1}, \boldsymbol{\psi}\right) + \hat{\mathbf{z}}_n\right\|_2^2 \right] < \infty.\end{align}

\item[3)] The function $\mathbf{f}\left(\hat{\boldsymbol{\psi}}_{n-1}, \boldsymbol{\psi}\right)$ should be continuous with respect to $\hat{\boldsymbol{\psi}}_{n-1}$.\\
By using \eqref{eq:est_ctrl}, we have 
\begin{align}
\!\!\!\!\mathbf{W}_n^\text{H}\mathbf{a}(x) \!=\! \left[\begin{matrix} \sum_{m=1}^M \frac{1}{\sqrt{M}}e^{-j \frac{2\pi d}{\lambda} (m-1) (x - \hat{x}_{n-1} + \delta^*)} \\ \sum_{m=1}^M \frac{1}{\sqrt{M}}e^{- j \frac{2\pi d}{\lambda} (m-1) (x - \hat{x}_{n-1} - \delta^*)} \end{matrix}\right]\!\!.\!\!\!
\end{align}
Since $e^{-j \frac{2\pi d}{\lambda} (m-1) (x - \hat{x}_{n-1} \pm \delta^*)}$ is continuous with respect to $\hat{x}_{n-1}$, and $\mathbf{W}_n^\text{H}\mathbf{a}(x)$ is the summation of a finite amount of $e^{-j \frac{2\pi d}{\lambda} (m-1) (x - \hat{x}_{n-1} \pm \delta^*)}, m = 1,\ldots, M$, we can get that $\mathbf{W}_n^\text{H}\mathbf{a}(x)$ is continuous with respect to $\hat{\boldsymbol{\psi}}_{n-1} \!=\! [\hat{\beta}^\text{re}_{n-1}, \hat{\beta}^\text{im}_{n-1}, \hat{x}_{n-1}]^\text{T}$. Similarly, we can prove that $\hat{\mathbf{g}}_{n}\!=\!\mathbf{W}_n^\text{H}\mathbf{a}(\hat{x}_{n\!-\!1})$, $\hat{\mathbf{e}}_{n}\!=\!\hat{\beta}_{n\!-\!1}\mathbf{W}_n^\text{H}\dot{\mathbf{a}}(\hat{x}_{n\!-\!1})$, $l_n = \|\hat{\mathbf{g}}_{n}\|_2 \|\hat{\mathbf{e}}_{n} \|_2$, and $c_n = \hat{\mathbf{g}}_{n}^\text{H}\hat{\mathbf{e}}_{n}$ are all continuous with respect to $\hat{\boldsymbol{\psi}}_{n-1}$.  

From \eqref{eq:function_f}, it can be observed that $\mathbf{f}\left(\hat{\boldsymbol{\psi}}_{n-1}, \boldsymbol{\psi}\right)$ is composed of finite numbers of $\mathbf{W}_n^\text{H}\mathbf{a}(x_n), \hat{\mathbf{g}}_{n}, \hat{\mathbf{e}}_{n}, l_n$ and $c_n$. Hence, we can conclude that $\mathbf{f}\left(\hat{\boldsymbol{\psi}}_{n-1}, \boldsymbol{\psi}\right)$ is continuous with respect to $\hat{\boldsymbol{\psi}}_{n-1}$.

\item[4)] Let $\boldsymbol\gamma_n = \mathbb{E} \left[ \left.\mathbf{f}\left(\hat{\boldsymbol{\psi}}_{n-1}, \boldsymbol{\psi}\right) + \hat{\mathbf{z}}_n\right| \mathcal{G}_{n-1} \right] - \mathbf{f}\left(\hat{\boldsymbol{\psi}}_{n-1}, \boldsymbol{\psi}\right)$. We need to prove that $\sum_{n=1}^\infty \left\| a_n \boldsymbol\gamma_n \right\|_2 < \infty$ with probability one. \\
From (\ref{eq_fil2}), we get $\boldsymbol\gamma_n = \mathbf{0}$ for all $n \ge 1$. So we have $\sum_{n=1}^\infty \left\| a_n \boldsymbol\gamma_n \right\|_2 = 0 < \infty$ with probability one.

% \item[5)] The set of stable points for the ODE \eqref{eq_ODE} should be obtained. \\
% According to \eqref{eq:stable_points}, $\mathcal{S}(x)$ contains the local optimal stable points for the ODE \eqref{eq_ODE}. What's more, the boundary point 1 (or $-1$) is a stable point when $f(1, x) \ge 0$ (or $f(-1, x) \le 0$). Hence, the set of stable points is a subset of $\mathcal{S}(x) \cup \{ -1\} \cup \{1 \}$.
\end{itemize}

By Theorem 5.2.1 in \cite{kushner2003stochastic}, $\hat{x}_n$ converges to a unique stable point within the stable points set $\mathcal{S}$ with probability one. % This completes the proof of Theorem~\ref{th_convergence}.

\section{Proof of Theorem~\ref{th:lock}}\label{proof:lock}

Theorem \ref{th:lock} is proven in three steps:

\emph{\textbf{Step 1:} We will construct two continuous processes based on the discrete process $ \hat{\boldsymbol{\psi}}_n = [\hat{\beta}^\text{re}_{n}, \hat{\beta}^\text{im}_{n}, \hat{x}_{n}]^\text{\emph{T}}$, i.e., $\bar{\boldsymbol{\psi}}(t) \!\overset{\Delta}{=}\! [\bar{\beta}^\text{re}(t), \bar{\beta}^\text{im}(t), \bar{x}(t)]^\text{\emph{T}}$ and $\tilde{\boldsymbol{\psi}}^n(t) \!\overset{\Delta}{=}\! [\tilde{\beta}^{\text{re},n}(t), \tilde{\beta}^{\text{im},n}(t), \tilde{x}^n(t)]^\text{\emph{T}}$.}  

Define the discrete time parameters: $t_{0} \overset{\Delta}{=} 0$, $t_n \overset{\Delta}{=} \sum_{i=1}^n a_{i}$, $n \ge 1$. The first continuous process $\bar{\boldsymbol{\psi}}(t), t \ge 0$ is the linear interpolation of the sequence $\hat{\boldsymbol{\psi}}_n, n \ge 0$, where $\bar{\boldsymbol{\psi}}(t_n) = \hat{\boldsymbol{\psi}}_n, n \ge 0$ and $\bar{\boldsymbol{\psi}}(t)$ is given by
\begin{equation}\label{eq_continuous}
\begin{aligned}
	\bar{\boldsymbol{\psi}}(t)\!=\!\bar{\boldsymbol{\psi}}(t_n)\!+\!\frac{(t\!-\!t_n)}{a_{n+1}}\left[\bar{\boldsymbol{\psi}}(t_{n+1})\!-\!\bar{\boldsymbol{\psi}}(t_n)\right], t\!\in\![t_n, t_{n+1}].
\end{aligned}
\end{equation}

The second continuous process $\tilde{\boldsymbol{\psi}}^n(t)$ is a solution of the following ordinary differential equation (ODE):
\begin{align}\label{eq_ODE}
\frac{d \tilde{\boldsymbol{\psi}}^n(t)}{dt} = \mathbf{f}\left(\tilde{\boldsymbol{\psi}}^n(t), \boldsymbol{\psi}\right),
\end{align}
for $t \in [t_n, \infty)$, where $\tilde{\boldsymbol{\psi}}^n(t_n) = \bar{\boldsymbol{\psi}}(t_n) = \hat{\boldsymbol{\psi}}_n, n \ge 0$. Hence, we have 
\begin{equation}\label{eq_ODE_new}
\begin{aligned}
	\tilde{\boldsymbol{\psi}}^n(t) & = \bar{\boldsymbol{\psi}}(t_n) + \int_{t_n}^t \mathbf{f}\left(\tilde{\boldsymbol{\psi}}^n(v), \boldsymbol{\psi}\right) dv, t \ge t_n.
\end{aligned}
\end{equation}

\emph{\textbf{Step 2:} By using the continuous processes $\bar{\boldsymbol{\psi}}(t)$ and $\tilde{\boldsymbol{\psi}}^n(t)$, we will form a sufficient condition for the convergence of the discrete process $\hat{x}_n$.}

We first construct a time-invariant set $\mathcal{I}$ that contains the real direction $x$ within the mainlobe, i.e., $x \in \mathcal{I} \subset \mathcal{B}(x)$. Pick $\delta$ such that\footnote{The boundary of the set $\mathcal{B}(x)$ is denoted by $\partial \mathcal{B}(x)$.}
\begin{equation}\label{eq:x_b}
\inf_{v \in \partial \mathcal{B}(x), t \ge 0} \left| v - \tilde{x}^0(t) \right| = \inf_{v \in \partial \mathcal{B}(x)} \left| v - \hat{x}_\text{b} \right| > \delta > 0,
\end{equation}
where $\hat{x}_\text{b} = \tilde{x}^0(t_\text{b})$ is the beam direction of the process $\tilde{\boldsymbol{\psi}}^0(t)$ that is closest to the boundary of the mainlobe (see e.g., Fig. \ref{fig_invariant_set}). Note that when $t \ge t_b$, the solution $\tilde{\boldsymbol{\psi}}^0(t)$ of the ODE (\ref{eq_ODE}) will approach the real channel coefficient $\beta$ and beam direction $x$ monotonically as time $t$ increases.
Hence, the invariant set $\mathcal{I}$ can be constructed as follows:
\begin{align}
	\mathcal{I} = \Big( x - |x - \hat{x}_\text{b}| - \delta,~x + |x - \hat{x}_\text{b}| + \delta \Big) \subset \mathcal{B}(x).
\end{align}
An example of the invariant set $\mathcal{I}$ is illustrated in Fig. \ref{fig_invariant_set}.

\begin{figure}[!t]
%\vspace{-5mm}
\centering
\includegraphics[width=4.5cm]{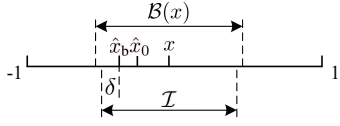}
%\vspace{-3mm}
\caption{An illustration of the invariant set $\mathcal{I}$.}
\label{fig_invariant_set}
%\vspace{-5mm}
\end{figure}

Then, we will establish a sufficient condition in Lemma \ref{le_sufficient} that ensures $\hat{x}_n\!\in\!\mathcal{I}~\text{for}~n\!\ge\!0$, and hence from Corollary 2.5 in \cite{borkar2008stochastic}, we can obtain that $\{\hat{x}_n\}$ converges to $x$.
Before giving Lemma \ref{le_sufficient}, let us provide some definitions first:
\begin{itemize}
\item Pick $T > 0$ such that the solution $\tilde{\boldsymbol{\psi}}^0(t), t\ge 0$ of the ODE (\ref{eq_ODE}) with $\tilde{\boldsymbol{\psi}}^0(0)\!=\![\hat{\beta}^\text{re}_{0}, \hat{\beta}^\text{im}_{0}, \hat{x}_{0}]^\text{T}$ satisfies $\inf_{v \in \partial \mathcal{B}}\left| v\!-\!\tilde{x}^0(t) \right| \ge 2\delta$ for $t \ge T$. Since when $t \ge t_b$, $\tilde{x}^0(t)$ will approach the real beam direction $x$ monotonically as time $t$ increases, one possible $T$ is given by
\begin{align}\label{eq_T}
T= \arg\min\limits_{t\in[t_\text{b}, \infty]}\left|~\!\!\left|\!\left[\int_{t_\text{b}}^t \mathbf{f}\left(\tilde{\boldsymbol{\psi}}^0(v), \boldsymbol{\psi}\right) dv\right]_3\right| - \delta\right|,
\end{align}
where $[\cdot]_{i}$ obtains the $i$-th element of the vector.

\item Let $T_0 \overset{\Delta}{=} 0$ and $T_{m+1} \overset{\Delta}{=} \min \left\{ t_i: t_i \ge T_n + T, i \ge 0 \right\}$ for $m \ge 0$. Then $T_{m+1} - T_m \in [T, T+a_1]$ and $T_m = t_{\tilde{n}(m)}$ for some $\tilde{n}(m) \uparrow \infty$, where $\tilde{n}(0) = 0$. Let $\tilde{\boldsymbol{\psi}}^{\tilde{n}(m)}(t)$ denote the solution of ODE (\ref{eq_ODE}) for $t \in I_m \overset{\Delta}{=} \left[ T_m, T_{m+1} \right]$ with $\tilde{\boldsymbol{\psi}}^{\tilde{n}(m)}(T_m) = \bar{\boldsymbol{\psi}}(T_m)$, $m \ge 0$.
\end{itemize}
Hence, we can obtain the following lemma:
\begin{lemma}\label{le_sufficient}
If $ \underset{t\in I_m}{\sup} \left| \bar{x}(t) - \tilde{x}^{\tilde{n}(m)}(t)\right| \le \delta$ for all $m \ge 0$, then $\hat{x}_n \in \mathcal{I}~\text{for all}~n \ge 0$.
\end{lemma}
\begin{proof}
See Appendix \ref{sec_proof_le_sufficient}
% The proof is omitted due to space limitation.
\end{proof}

\emph{\textbf{Step 3:} We will derive the probability lower bound for the condition in Lemma \ref{le_sufficient}, which is also a lower bound for $P\left( \left. \hat{x}_n\!\rightarrow\!x \right| \hat{x}_0\!\in\!\mathcal{B}\left(x\right) \right)$. }

We will derive the probability lower bound for the condition in Lemma \ref{le_sufficient}, which results in the following lemma:
\begin{lemma}\label{le_lower_bound}
If (i) the initial point satisfies $\hat{x}_0 \in \mathcal{B}(x)$, (ii) $a_n$ is given by (\ref{eq:stepsize}) with any $\alpha > 0$,
then there exist  $N_0 \ge 0$ and $C>0$ such that
\begin{equation}\label{eq_lock}
\begin{aligned}
P\left( \hat{x}_n \in \mathcal{I}, \forall n \ge 0 \right) \ge 1 - 6e^{-\frac{C|s|^2}{\alpha^2\sigma_0^2}}.
\end{aligned}
\end{equation}
\end{lemma}
\begin{proof} See Appendix \ref{sec_proof_le_lower_bound}.
\end{proof}

Finally, by applying Lemma \ref{le_lower_bound} and Corollary 2.5 in \cite{borkar2008stochastic}, we can obtain
\begin{align}\label{eq_lock10}
	P\left( \left. \hat{x}_n \rightarrow x \right| \hat{x}_0 \in \mathcal{B} \right) \ge &~ P\left( \hat{x}_n \in \mathcal{I}, \forall n \ge 0 \right) \\
	\ge &~ 1 - 6e^{-\frac{C|{s}|^2}{\alpha^2\sigma_0^2}},\nonumber
\end{align}
which completes the proof of Theorem \ref{th:lock}.

\section{Proof of Theorem~\ref{th:normal}}\label{proof:normal}

When the step-size $a_n$ is given by (\ref{eq:stepsize}) with any $\alpha > 0$ and $N_0 \ge 0$, Theorem 6.6.1 \cite[Section 6.6]{nevel1973stochastic} has proposed the sufficient conditions to prove the asymptotic normality of $\sqrt{n} \left( \hat{x}_n - x \right)$, i.e., $\sqrt{n} \left( \hat{x}_n - x \right) \overset{d}{\rightarrow}\mathcal{N}\left( 0, \Sigma_x \right)$. Under the condition that $\hat{\boldsymbol{\psi}}_n \rightarrow \boldsymbol{\psi}$, we will prove that our algorithm satisfies its sufficient conditions and obtain the variance $\Sigma$ as follows:
\begin{itemize}
% \item[1)] The estimate $\hat{x}_n$ should be within $[-1, 1]$.\\
% The projection operator in \eqref{eq_est} ensures that $\hat{x}_n \in [-1, 1]$.

\item[1)] Equation \eqref{eq:newton4} should satisfy: (i) there exist an increasing sequence of $\sigma$-fields $\{\mathcal{F}_{n}: n \ge 0\}$ such that $\mathcal{F}_{m} \!\subset\!\mathcal{F}_{n}$ for $m\!<\!n$, and (ii) the random noise $\hat{\mathbf{z}}_n$ is $\mathcal{F}_{n}$-measurable and independent of $\mathcal{F}_{n-1}$.\\
As defined in Appendix \ref{proof:convergence}, there exist an increasing sequence of $\sigma$-fields $\{ \mathcal{G}_n : n \ge 0 \}$, such that $\hat{\mathbf{z}}_n$ is measurable with respect to $\mathcal{G}_{n}$, i.e., $\mathbb{E} \left[ \left. \hat{\mathbf{z}}_n \right| \mathcal{G}_{n} \right] = \hat{\mathbf{z}}_n$, and is independent of $\mathcal{G}_{n-1}$, i.e., $\mathbb{E} \left[ \left. \hat{\mathbf{z}}_n \right| \mathcal{G}_{n-1}\right]  = \mathbb{E} \left[ \hat{\mathbf{z}}_n \right] = \mathbf{0}$.

\item[2)] $\hat{x}_n$ should converge to $x$ almost surely as $n \rightarrow \infty$. \\
Since $\hat{\boldsymbol{\psi}}_n \rightarrow \boldsymbol{\psi}$ is assumed, we have that $\hat{x}_n$ converges to $x$ almost surely as $n \rightarrow \infty$.

\item[3)] The stable condition:\\
In \eqref{eq:function_f}, $\mathbf{f}\left(\hat{\boldsymbol{\psi}}_{n-1}, \boldsymbol{\psi}\right)$ can be rewritten as follows:
\begin{equation}
\begin{aligned}
\!\!\!\!\!\!\mathbf{f}\left(\hat{\boldsymbol{\psi}}_{n-1}, \boldsymbol{\psi}\right)\!=\!\mathbf{C}_1 \left( \hat{\boldsymbol{\psi}}_{n-1} - \boldsymbol{\psi} \right) \!+\! \left[\begin{matrix} o(\| \hat{\boldsymbol{\psi}}_{n-1} - \boldsymbol{\psi} \|_2) \\ o(\| \hat{\boldsymbol{\psi}}_{n-1} - \boldsymbol{\psi} \|_2) \\ o(\| \hat{\boldsymbol{\psi}}_{n-1} - \boldsymbol{\psi} \|_2) \end{matrix}\right]\!\!,\!\!\!\!
%& = -\frac{M(M-1)^2|\beta x|^2}{2\sigma^2} \left( \hat{x}_n - u^* \right) + o\left(\hat{x}_n - u^*\right).
\end{aligned}
\end{equation}
where $\mathbf{C}_1$ is given by
\begin{equation}
\begin{aligned}
\mathbf{C}_1 =  \left.\frac{\partial \mathbf{f}\left(\hat{\boldsymbol{\psi}}_{n-1}, \boldsymbol{\psi}\right)}{\partial \hat{\boldsymbol{\psi}}_{n-1}^\text{T}} \right|_{\hat{\boldsymbol{\psi}}_{n-1} = \boldsymbol{\psi}} = - \left[\begin{matrix} 1 & 0 & 0 \\ 0 & 1 & 0 \\ 0 & 0 & 1 \end{matrix} \right] .
\end{aligned}
\end{equation}
Then we get the stable condition that
\begin{align}
\mathbf{A} \!=\! \mathbf{C}_1  \cdot \alpha + \frac{1}{2} \!=\! - \!\left[\begin{matrix} \alpha \!-\! \frac{1}{2} & 0 & 0 \\ 0 & \alpha \!-\! \frac{1}{2} & 0 \\ 0 & 0 & \alpha \!-\! \frac{1}{2} \end{matrix} \right]  \prec 0,
\end{align}
which results in $\alpha > \frac{1}{2}$.

\item[4)] The constraints for the noise vector $\hat{\mathbf{z}}_n$:\\
\begin{equation}
\mathbb{E}\left[\left\|\hat{\mathbf{z}}_n\right\|_2^2\right] = \operatorname{tr}(\mathbf{I}(\hat{\boldsymbol{\psi}}_{n\!-\!1},\!\mathbf{W}_n)^{-1}) < \infty,
\end{equation}
and
\begin{equation}
\underset{v\rightarrow\infty}{\lim}\ \ \underset{n\ge1}{\sup}\ \ \int\limits_{\left\| \hat{z}_n \right\|_2 > v} \left\| \hat{\mathbf{z}}_n \right\|_2^2 p(\hat{\mathbf{z}}_n) d\hat{\mathbf{z}}_n = 0.
\end{equation}

\end{itemize}

Let 
\begin{align}
\mathbf{B} = &~ \lim_{\begin{matrix} n \rightarrow \infty \\ \hat{\boldsymbol{\psi}}_n \rightarrow \boldsymbol{\psi} \end{matrix}} \mathbb{E}\left[\hat{\mathbf{z}}_n \hat{\mathbf{z}}_n^\text{T}\right] \\
 \overset{(a)}{=} &~ \lim_{\begin{matrix} n \rightarrow \infty \\ \hat{\boldsymbol{\psi}}_n \rightarrow \boldsymbol{\psi} \end{matrix}}\mathbf{I}(\hat{\boldsymbol{\psi}}_{n},\!\mathbf{W}_{n\!+\!1})^{-1} = \mathbf{I}(\boldsymbol{\psi}, \mathbf{W}^*)^{-1}, \nonumber
\end{align}
where step $(a)$ is obtained from \eqref{eq:covariance}.

Then, from Theorem 6.6.1 \cite[Section 6.6]{nevel1973stochastic}, we have
\begin{align*}\sqrt{n+N_0}\left( \hat{\boldsymbol{\psi}}_{n} - \boldsymbol{\psi} \right) \overset{d}{\rightarrow}\mathcal{N}\left( 0, \boldsymbol\Sigma \right), \end{align*}
where
\begin{equation}\label{eq_Sigma}\begin{aligned}
\boldsymbol\Sigma = &~\alpha^2  \cdot \int_0^\infty e^{\mathbf{A}v} \mathbf{B} e^{\mathbf{A}^\text{H}v} dv \\
= &~\frac{\alpha^2}{2\alpha - 1}\mathbf{I}(\boldsymbol{\psi}, \mathbf{W}^*)^{-1}.\\ %& = \left( \frac{\lambda^2 }{2M(M-1)^2\pi^2 d^2\rho} \right)^2 *  \frac{2M(M-1)^2\pi^2 d^2\rho}{\lambda^2} \\
%& \ge I_{\max}^{-1} = \frac{\lambda^2 }{2M(M-1)^2\pi^2 d^2\rho}.
\end{aligned}\end{equation}
Due to that $\lim_{n\rightarrow\infty}\sqrt{{(n+N_0)}/{n}} = 1$, we have
\begin{align*}
\sqrt{n}\left( \hat{\boldsymbol{\psi}}_{n} - \boldsymbol{\psi} \right) \rightarrow \sqrt{n}\cdot\sqrt{\frac{n+N_0}{n}}\left( \hat{\boldsymbol{\psi}}_{n} - \boldsymbol{\psi} \right) \overset{d}{\rightarrow}\mathcal{N}\left( 0, \boldsymbol\Sigma \right),
\end{align*}
as $n\rightarrow\infty$. Hence, we can obtain 
\begin{align}
\sqrt{n} \left( \hat{x}_n - x \right) \overset{d}{\rightarrow}\mathcal{N}\left( 0, \left[\boldsymbol\Sigma\right]_{3,3} \right).
\end{align}
By adapting $\alpha$ in \eqref{eq_Sigma}, we can obtain different $\left[\boldsymbol\Sigma\right]_{3,3}$, which achieves the minimum value $\left[\mathbf{I}(\boldsymbol{\psi}, \mathbf{W}^*)^{-1}\right]_{3,3}$, i.e., the minimum CRLB in (\ref{eq:CRLB}), when $\alpha = 1$.

By assuming $\alpha = 1$, we conclude that
\begin{equation*}\label{eq_normal2}
	\lim_{n\rightarrow\infty}~n~\mathbb{E}\left[\left(\hat{x}_n - x\right)^2\big| \hat{\boldsymbol{\psi}}_n \rightarrow \boldsymbol{\psi}\right] = \left[\mathbf{I}(\boldsymbol{\psi}, \mathbf{W}^*)^{-1}\right]_{3,3}.
\end{equation*}

\section{Proof of Lemma \ref{le_sufficient}}\label{sec_proof_le_sufficient}

When $m = 0$, $\tilde{x}^{\tilde{n}(0)}(T_0) = \bar{x}(T_0) = \hat{x}_0$. There are two symmetrical cases: (i) $\hat{x}_0 < x$ and (ii) $\hat{x}_0 > x$.

\emph{Case 1} ($\hat{x}_0 < x$): We will first prove that $\bar{x}(t) \in \mathcal{I} = \Big( x - |x - \hat{x}_b| - \delta,~x + |x - \hat{x}_b| + \delta \Big)$ for all $t \in I_0$.

If $\left| \bar{x}(t) - \tilde{x}^{\tilde{n}(0)}(t)\right| \le \delta$ for all $t \in I_0$, then we have
\begin{equation}\label{eq:le_sufficient_1}
-\delta \le \bar{x}(t) - \tilde{x}^{\tilde{n}(0)}(t) \le \delta.
\end{equation} 
What's more, due to the definition of $\hat{x}_\text{b}$ in \eqref{eq:x_b}, we get 
\begin{equation}\label{eq:le_sufficient_2}
\hat{x}_\text{b} \le x, \tilde{x}^{\tilde{n}(0)}(t) - \hat{x}_\text{b} \ge 0, x - \tilde{x}^{\tilde{n}(0)}(t) \ge 0,
\end{equation}
for all $t \in I_0$. By using \eqref{eq:le_sufficient_1} and \eqref{eq:le_sufficient_2}, we can obtain
\begin{align}\label{eq_left}
	&~ \bar{x}(t) - (x - |x - \hat{x}_b| - \delta)\\
	= &~\bar{x}(t) - (\hat{x}_\text{b} - \delta) \nonumber\\
	= &~\left[\bar{x}(t) - \tilde{x}^{\tilde{n}(0)}(t)\right] + \left[\tilde{x}^{\tilde{n}(0)}(t) - \hat{x}_\text{b}\right] + \delta \ge 0, \nonumber % \ge &~ -\delta + 0  + \delta = 0,\nonumber
\end{align}
and
\begin{align}\label{eq_right}
	&~(x + |x - \hat{x}_\text{b}| + \delta) - \bar{x}(t) \\
	= &~(2x - \hat{x}_\text{b} + \delta) - \bar{x}(t) \nonumber\\	
	= &~\left(x - \hat{x}_\text{b}\right) + \left[x - \bar{x}(t)\right] + \delta \nonumber\\
	= &~\left(x - \hat{x}_\text{b}\right) + \left[x - \tilde{x}^{\tilde{n}(0)}(t)\right] + \left[\tilde{x}^{\tilde{n}(0)}(t) - \bar{x}(t)\right] + \delta \nonumber\\
	\ge &~ 0,\nonumber
\end{align}
which result in $\bar{x}(t) \in \mathcal{I}$ for all $t \in I_0$.

Then, we consider the initial value $\bar{x}(T_1)$ for the next time interval $I_1$. With the $T$ given by \eqref{eq_T}, we have
\begin{align}\label{eq:le_sufficient_3}
x - \hat{x}_\text{b} \ge \tilde{x}^{\tilde{n}(0)}(T_1) - \hat{x}_\text{b} \ge \tilde{x}^{\tilde{n}(0)}(T) - \hat{x}_\text{b} \ge \delta.
\end{align}
By using \eqref{eq:le_sufficient_1}, \eqref{eq:le_sufficient_2} and \eqref{eq:le_sufficient_3}, we get
\begin{align}\label{eq_left2}
	&~\bar{x}(T_1) - (x - |x - \hat{x}_\text{b}|) \\
	= &~\bar{x}(T_1) - \hat{x}_\text{b} \nonumber \\
	= &~\left[\bar{x}(T_1) - \tilde{x}^{\tilde{n}(0)}(T_1)\right] + \left[\tilde{x}^{\tilde{n}(0)}(T_1) - \hat{x}_\text{b}\right] \ge 0,\nonumber
\end{align}
and
\begin{align}\label{eq_right2}
	&~(x + |x - \hat{x}_\text{b}|) - \bar{x}(T_1) \\
	= &~(2x - \hat{x}_\text{b}) - \bar{x}(T_1) \nonumber\\	
	= &~\left(x - \hat{x}_\text{b}\right) + \left[x - \bar{x}(T_1)\right] \nonumber\\
	= &~\left(x - \hat{x}_\text{b}\right) + \left[x - \tilde{x}^{\tilde{n}(0)}(T_1)\right] + \left[\tilde{x}^{\tilde{n}(0)}(T_1) - \bar{x}(T_1)\right] \nonumber\\
	\ge &~0,\nonumber
\end{align}
which result in $\bar{x}(T_1) \in \big[ x - |x - \hat{x}_\text{b}|,~x + |x - \hat{x}_\text{b}| \big]$.

\emph{Case 2} ($\hat{x}_0 > x$): Owing to symmetric property, we can use the same method as \eqref{eq_left}, \eqref{eq_right}, \eqref{eq_left2} and \eqref{eq_right2} to obtain that $\bar{x}(t) \in \mathcal{I}$ for all $t \in I_0$ and $\bar{x}(T_1) \in \big[ x - |x - \hat{x}_\text{b}|,~x + |x - \hat{x}_\text{b}| \big]$.

%\vspace{3mm}
When $m = 1$, $\tilde{x}^{\tilde{n}(1)}(T_1) = \bar{x}(T_1) \in \big[ x - |x - \hat{x}_\text{b}|,~x + |x - \hat{x}_\text{b}| \big]$. If $\bar{x}(T_1) < x$ and $\left| \bar{x}(t) - \tilde{x}^{\tilde{n}(1)}(t)\right| \le \delta$, then for all $t \in I_1$, we have $\bar{x}(T_1) \ge  \hat{x}_\text{b}$, $\tilde{x}^{\tilde{n}(1)}(t) - \hat{x}_\text{b} \ge 0$, $x - \tilde{x}^{\tilde{n}(1)}(t) \ge 0$, and
\begin{align*}
x - \hat{x}_\text{b} \ge \tilde{x}^{\tilde{n}(1)}(T_2) - \hat{x}_\text{b} \ge \tilde{x}^{\tilde{n}(1)}(T_1+T) - \hat{x}_\text{b} \ge \delta.
\end{align*}
Similar to  \eqref{eq_left}, \eqref{eq_right}, \eqref{eq_left2} and \eqref{eq_right2}, we can get $\bar{x}(t) \in \mathcal{I}~\text{for all}~t \in I_1$ and $\bar{x}(T_2) \in \big[ x - |x - \hat{x}_\text{b}|,~x + |x - \hat{x}_\text{b}| \big]$, which are also true for the case that $\bar{x}(T_1) > x$.

%\vspace{3mm}
Hence, we can use the same method to prove the cases of $m \ge 2$, which finally yields $\bar{x}(t) \in \mathcal{I}$ for all $t \in I_m$ and $m \ge 0$. Since $\bar{x}(t_n) = \hat{x}_n$ for all $n \ge 0$, we can obtain that $\hat{x}_n \in \mathcal{I}~\text{for all}~n \ge 0$, which completes the proof.

\section{Proof of Lemma \ref{le_lower_bound}}\label{sec_proof_le_lower_bound}

The following lemmas are needed to prove Lemma \ref{le_lower_bound}:
\begin{lemma}\label{le_gronwall}
Given $T$ by \eqref{eq_T} and
\begin{align}\label{eq_nT}
n_T \overset{\Delta}{=} \inf \left\{i \in \mathbb{Z}: t_{n+i} \ge t_n + T \right\}.
\end{align}
If there exists a constant $C>0$, which satisfies
\begin{equation}\label{eq_gronwall1}
\begin{aligned}
	&~\left\| \bar{\boldsymbol{\psi}}(t_{n+m}) - \tilde{\boldsymbol{\psi}}^n(t_{n+m})\right\|_2 \\
	\le &~L \sum_{i=1}^{m} a_{n+i} \left\| \bar{\boldsymbol{\psi}}(t_{n+i-1}) - \tilde{\boldsymbol{\psi}}^n(t_{n+i-1}) \right\|_2  + C,
\end{aligned}
\end{equation}
for all $n \ge 0$ and $1 \le m \le n_T$, then
\begin{equation}\label{eq_gronwall2}
\begin{aligned}
	\underset{t\in\left[ t_n, t_{n+n_T} \right]}{\sup} \left\| \bar{\boldsymbol{\psi}}(t) - \tilde{\boldsymbol{\psi}}^n(t)\right\|_2 \le  \frac{C_{\mathbf{f}} a_{n+1}}{2} + C e^{L (T+a_1)},
\end{aligned}
\end{equation}
where $L$ and $C_{\mathbf{f}}$ are defined in \eqref{eq_Lip} and \eqref{eq_CT} separately.
\end{lemma}
\begin{proof}
See Appendix \ref{sec_proof_le_gronwall}.
% The proof is omitted due to space limitation.
\end{proof}

%\begin{lemma}\label{le_cheb}
%If a $D$-dimensional process 
%$$\{\mathbf{M}_n = [M_{n,1}, \ldots, M_{n,D}]^\text{T}: n = 1, 2, \ldots\},$$
%statisfies that: (i) $\mathbf{M}_n$ is Gaussian distributed with zero mean, i.e., $\mathbf{M}_n \sim \mathcal{N}(\mathbf{0}, \mathbf{V}_n)$, and (ii)~$\sum_{i = 1}^D \left(M_{n,i}^2 - \mathbb{E}\left[M_{n,i}^2\right]\right)$ is a martingale in $n$, then
%\begin{equation}\label{eq_lock5}
%\begin{aligned}
%	&~P\left( \underset{0\le n \le k}{\sup}\sum_{i = 1}^D \left(M_{n,i}^2 - \mathbb{E}\left[M_{n,i}^2\right]\right) > \eta \right) \\
%	\le &~\frac{e^{-C\left(\sum_{i=1}^D \lambda_{k,i} + \eta\right)} }{\sqrt{\prod_{i = 1}^D \left(1 - 2C\lambda_{k,i}\right)}},
%\end{aligned}
%\end{equation}
%for any $\eta > 0$, where $\lambda_{k,1} \le \lambda_{k,2} \le \cdots \le \lambda_{k,D}$ are the eigenvalues of $\mathbf{V}_k$, and $0 < C < \frac{1}{2\lambda_{k,D}}$.
%\end{lemma}
\begin{lemma}[Lemma 4 \cite{Li2017super}]\label{le_cheb}
If $\{M_i: i = 1, 2, \ldots\}$ satisfies that: (i)  $M_i$ is Gaussian distributed with zero mean, and (ii) $M_i$ is a martingale in $i$, then
\begin{equation}\label{eq_lock5}
\begin{aligned}
	& P\left( \underset{0\le i \le k}{\sup}\left|M_i\right| > \eta \right) \le 2\exp\left\{-\frac{\eta^2}{2\operatorname{Var}\left[M_k\right]}\right\},
\end{aligned}
\end{equation}
for any $\eta > 0$.
\end{lemma}
%\begin{proof}
%See Appendix \ref{sec_proof_le_cheb}.
%% The proof is omitted due to space limitation.
%\end{proof}

\begin{lemma}[Lemma 5 \cite{Li2017super}]\label{le_sum}
If given a constant $C > 0$, then
\begin{equation}\label{eq_increasing}
\begin{aligned}
	G(v) = \frac{1}{v}\exp\left[-\frac{C}{v}\right],
\end{aligned}
\end{equation}
is increasing for all $0 < v < C$.
\end{lemma}
%\begin{proof}
%The derivative of $G(v)$ is
%\begin{equation*}
%	G'(v) = \frac{C - v}{v^3}\exp\left[-\frac{C}{v}\right].
%\end{equation*}
%
%Let $G'(v) > 0$ and we can obtain that $G(v)$ is increasing for $v \in (0, C)$, which completes the proof.
%\end{proof}

Let $\boldsymbol{\xi}_0 \overset{\Delta}{=} \mathbf{0}$ and $\boldsymbol{\xi}_n \overset{\Delta}{=} \sum_{m=1}^{n} a_{m} \mathbf{\hat{z}}_{m} $, $n \ge 1$, where $\mathbf{\hat{z}}_{m}$ is given in \eqref{eq:noise_z2}. With \eqref{eq_continuous} and \eqref{eq_ODE_new}, we have for $t_{n+m}, 1 \le m \le n_T$,
\begin{align}\label{eq_seq_trace} \bar{\boldsymbol{\psi}}(t_{n+m}) = &~\bar{\boldsymbol{\psi}}(t_n)  + \sum_{i=1}^{m} a_{n+i} \mathbf{f}\left(\bar{\boldsymbol{\psi}}(t_{n+i-1}), \boldsymbol{\psi}\right) \\
&~+ (\boldsymbol{\xi}_{n+m} - \boldsymbol{\xi}_{n}), \nonumber
\end{align}
and
\begin{align}\label{eq_ode_trace}
\tilde{\boldsymbol{\psi}}^n(t_{n+m}) = &~\tilde{\boldsymbol{\psi}}^n(t_n) + \int_{t_n}^{t_{n+m}} \mathbf{f}\left(\tilde{\boldsymbol{\psi}}^n(v), \boldsymbol{\psi}\right) dv \\
= &~\tilde{\boldsymbol{\psi}}^n(t_n) + \sum_{i=1}^{m} a_{n+i} \mathbf{f}\left(\tilde{\boldsymbol{\psi}}^n(t_{n+i-1}), \boldsymbol{\psi}\right) \nonumber \\
	&~+ \int_{t_n}^{t_{n+m}} \left[\mathbf{f}\left(\tilde{\boldsymbol{\psi}}^n(v), \boldsymbol{\psi}\right) - \mathbf{f}\left(\tilde{\boldsymbol{\psi}}^n(\underline{v}), \boldsymbol{\psi}\right) \right]dv,  \nonumber
\end{align}
where $\underline{v} \overset{\Delta}{=} \max \left\{ t_n: t_n \le v, n \ge 0 \right\}$ for $v \ge 0$.
%Note that we only care about $\hat{x}_n \in \mathcal{I} \subset \mathcal{B}$, so the projection operator does not take effect in \eqref{eq_seq_trace} and we omit it.

To bound $\int_{t_n}^{t_{n+m}} \left[\mathbf{f}\left(\tilde{\boldsymbol{\psi}}^n(v), \boldsymbol{\psi}\right) - \mathbf{f}\left(\tilde{\boldsymbol{\psi}}^n(\underline{v}), \boldsymbol{\psi}\right) \right]dv$ on the RHS of (\ref{eq_ode_trace}), we obtain the Lipschitz constant of function $\mathbf{f}(\mathbf{v}, \boldsymbol{\psi})$ considering the first varible $\mathbf{v}$, given by
\begin{equation}\label{eq_Lip}
	L \overset{\Delta}{=} \underset{\mathbf{v}_1 \ne \mathbf{v}_2}{\sup} \frac{\left\| \mathbf{f}(\mathbf{v}_1, \boldsymbol{\psi}) -  \mathbf{f}(\mathbf{v}_2, \boldsymbol{\psi}) \right\|_2}{\left\| \mathbf{v}_1 - \mathbf{v}_2 \right\|_2}.
\end{equation}
Similar to \eqref{eq_ub_fx}, for any $t \ge t_n$, we can obtain that there exists a constant $0 < C_{\mathbf{f}} < \infty$ such that
\begin{equation}\label{eq_CT}
\begin{aligned}
	\left\| \mathbf{f}\left(\tilde{\boldsymbol{\psi}}^n(t), \boldsymbol{\psi}\right) \right\|_2 \le C_{\mathbf{f}}.
\end{aligned}
\end{equation}
Hence, we have
\begin{equation}\label{eq_int}
\begin{aligned}
& \left\| \int_{t_n}^{t_{n+m}} \left[\mathbf{f}\left(\tilde{\boldsymbol{\psi}}^n(v), \boldsymbol{\psi}\right) - \mathbf{f}\left(\tilde{\boldsymbol{\psi}}^n(\underline{v}), \boldsymbol{\psi}\right) \right]dv \right\|_2 \\
\le & \int_{t_n}^{t_{n+m}} \left\| \mathbf{f}\left(\tilde{\boldsymbol{\psi}}^n(v), \boldsymbol{\psi}\right) - \mathbf{f}\left(\tilde{\boldsymbol{\psi}}^n(\underline{v}), \boldsymbol{\psi}\right) \right\|_2 dv \\
\overset{(a)}{\le} & \int_{t_n}^{t_{n+m}} L \left\| \tilde{\boldsymbol{\psi}}^n(v) - \tilde{\boldsymbol{\psi}}^n(\underline{v}) \right\|_2 dv \\
\overset{(b)}{\le} & \int_{t_n}^{t_{n+m}} L \left\| \int_{\underline{v}}^{v} \mathbf{f}\left(\tilde{\boldsymbol{\psi}}^n(s), \boldsymbol{\psi}\right) ds \right\|_2 dv \\
\le & \int_{t_n}^{t_{n+m}} \int_{\underline{v}}^{v} L \left\| \mathbf{f}\left(\tilde{\boldsymbol{\psi}}^n(s), \boldsymbol{\psi}\right) \right\|_2 ds dv \\
\overset{(c)}{\le} & \int_{t_n}^{t_{n+m}} \int_{\underline{v}}^{v} C_{\mathbf{f}} L ds dv =   \int_{t_n}^{t_{n+m}}C_{\mathbf{f}} L (v - \underline{v}) dv \\
=  & \sum_{i=1}^{m} \int_{t_{n+i-1}}^{t_{n+i}}C_{\mathbf{f}} L (v - t_{n+i-1}) dv \\
= & \sum_{i=1}^{m} \frac{C_{\mathbf{f}} L (t_{n+i} - t_{n+i-1})^2}{2}= \frac{C_{\mathbf{f}} L}{2} \sum_{i=1}^{m} a_{n+i}^2,
\end{aligned}
\end{equation}
where step $(a)$ is due to (\ref{eq_Lip}), step $(b)$ is due to the definition in (\ref{eq_ODE_new}), and step $(c)$ is due to (\ref{eq_CT}). Then, by subtracting $\tilde{\boldsymbol{\psi}}^n(t_{n+m})$ in \eqref{eq_ode_trace} from $\bar{\boldsymbol{\psi}}(t_{n+m})$  in \eqref{eq_seq_trace} and taking norms, the following inequality can be obtained from (\ref{eq_Lip}) and (\ref{eq_int}) for $n \ge 0, 1 \le m \le n_T$:
\begin{equation}\label{eq_lock2}
\begin{aligned}
	& \left\| \bar{\boldsymbol{\psi}}(t_{n+m}) - \tilde{\boldsymbol{\psi}}^n(t_{n+m})\right\|_2 \\
	\le & L \sum_{i=1}^{m} a_{n+i} \left\| \bar{\boldsymbol{\psi}}(t_{n+i-1}) - \tilde{\boldsymbol{\psi}}^n(t_{n+i-1}) \right\|_2 \\
	&  + \frac{C_{\mathbf{f}} L}{2} \sum_{i=1}^{m} a_{n+i}^2+ \left\|\boldsymbol{\xi}_{n+m} - \boldsymbol{\xi}_{n}\right\|_2 \\
	\le & L \sum_{i=1}^{m} a_{n+i} \left\| \bar{\boldsymbol{\psi}}(t_{n+i-1}) - \tilde{\boldsymbol{\psi}}^n(t_{n+i-1}) \right\|_2 \\
	&  + \frac{C_{\mathbf{f}} L}{2} \sum_{i=1}^{n_T} a_{n+i}^2+ \underset{1 \le m\le n_T}{\sup}\left\|\boldsymbol{\xi}_{n+m} - \boldsymbol{\xi}_{n}\right\|_2.
%	& \left| \bar{x}(t_{n+m}) - \tilde{x}^n(t_{n+m})\right| \\
%	\le & L \sum_{i=1}^{m} a_{n+i} \left| \bar{x}(t_{n+i-1}) - \tilde{x}^n(t_{n+i-1}) \right| \\
%	&  + \frac{\sqrt{M} L}{2} \sum_{i=1}^{m} a_{n+i}^2+ |\xi_{n+m} - \xi_{n}|.
\end{aligned}
\end{equation}

Applying Lemma \ref{le_gronwall} to (\ref{eq_lock2}) and letting
\begin{align*}C = \frac{C_{\mathbf{f}} L}{2} \sum_{i=1}^{n_T} a_{n+i}^2+ \underset{1\le m\le n_T}{\sup}\left\|\boldsymbol{\xi}_{n+m} - \boldsymbol{\xi}_{n}\right\|_2,\end{align*}
yields
\begin{equation}\label{eq_lock3}
\begin{aligned}
	&~\underset{t\in\left[ t_n, t_{n+n_T} \right]}{\sup} \left\| \bar{\boldsymbol{\psi}}(t) - \tilde{\boldsymbol{\psi}}^n(t)\right\|_2	 \\
	\le &~C_e \left\{ \frac{C_{\mathbf{f}} L}{2} \big[b(n) - b(n+n_T)\big] \right.  \\
	& \left. + \underset{1 \le m\le n_T}{\sup}\left\|\boldsymbol{\xi}_{n+m} - \boldsymbol{\xi}_{n}\right\|_2 \right\} + \frac{C_{\mathbf{f}} a_{n+1}}{2} ,
\end{aligned}
\end{equation}
where $C_e \overset{\Delta}{=} e^{L (T+a_1)}$, and $b(n) \overset{\Delta}{=} \sum_{i > n} a_{i}^2$.
% Suppose the initial point $\hat{x}_0 \in \mathcal{B}$.
% Then, we will derive the lower bound of the probability that the sequence $\{\hat{x}_n: n \ge 0\}$ remains inside this invariant set.
% Let $T_0 \overset{\Delta}{=} 0$ and $T_{m+1} \overset{\Delta}{=} \min \left\{ t_i: t_i \ge T_n + T, i \ge 0 \right\}$ for $m \ge 0$. Then $T_{m+1} - T_m \in [T, T+a_1]$ and $T_m = t_{\tilde{n}(m)}$ for some $\tilde{n}(m) \uparrow \infty$, where $\tilde{n}(0) = n_0$. Let $\tilde{x}^{\tilde{n}(m)}(t)$ denote the solution of ODE (\ref{eq_ODE}) for $t \in I_m \overset{\Delta}{=} \left[ T_m, T_{m+1} \right]$ with $\tilde{x}^{\tilde{n}(m)}(T_m) = \bar{x}(T_m)$, $m \ge 0$.
Letting $n = \tilde{n}(m)$ in (\ref{eq_lock3}), we have $n + n_T = \tilde{n}(m+1)$ due to the definition of $T_{m+1} = t_{\tilde{n}(m+1)}$ in \emph{Step 2} of Appendix \ref{proof:lock} and
\begin{equation}\label{eq_lock3-2}
\begin{aligned}
	&~\underset{t\in I_m}{\sup} \left\| \bar{\boldsymbol{\psi}}(t) - \tilde{\boldsymbol{\psi}}^{\tilde{n}(m)}(t)\right\|_2	 \\
	\le &~C_e \left\{ \frac{C_{\mathbf{f}} L}{2} \big[ b(\tilde{n}(m)) - b(\tilde{n}(m+1)) \big]\right. \\
	& \left.  + \underset{\tilde{n}(m) \le k \le \tilde{n}(m+1)}{\sup}\left\|\boldsymbol{\xi}_{k} - \boldsymbol{\xi}_{\tilde{n}(m)}\right\|_2 \right\} + \frac{C_{\mathbf{f}} a_{\tilde{n}(m)+1}}{2}.
\end{aligned}
\end{equation}

Suppose that the step size $\{a_n: n > 0\}$ satisfies
\begin{equation}\label{eq_lock_constr}
C_e \frac{C_{\mathbf{f}} L}{2} \big[b(\tilde{n}(m)) - b(\tilde{n}(m+1))\big] +  \frac{C_{\mathbf{f}} a_{\tilde{n}(m)+1}}{2} < \frac{\delta}{2},
\end{equation}
for $m \ge 0$.

Given $\underset{t\in I_m}{\sup} \left| \bar{x}(t) - \tilde{x}^{\tilde{n}(m)}(t)\right| \!>\! \delta$, we can obtain from \eqref{eq_lock3-2} and \eqref{eq_lock_constr} that
\begin{equation*}
\begin{aligned}
	&~\underset{\tilde{n}(m)\le k \le \tilde{n}(m+1)}{\sup}\left\|\boldsymbol{\xi}_{k} - \boldsymbol{\xi}_{\tilde{n}(m)}\right\|_2 \\
	\ge &~\frac{1}{C_e} \left( \underset{t\in I_m}{\sup} \left\| \bar{\boldsymbol{\psi}}(t) - \tilde{\boldsymbol{\psi}}^{\tilde{n}(m)}(t)\right\|_2 - \frac{C_{\mathbf{f}} L}{2} \big[ b(\tilde{n}(m)) \right. \\
		& \left.  - b(\tilde{n}(m+1)) \big] - \frac{C_{\mathbf{f}} a_{\tilde{n}(m)+1}}{2}\right) \\
		> &~\frac{1}{C_e}\left( \underset{t\in I_m}{\sup} \left| \bar{x}(t) - \tilde{x}^{\tilde{n}(m)}(t)\right| - \frac{\delta}{2} \right) \\
		> &~\frac{\delta}{2C_e}.
\end{aligned}
\end{equation*}
Then, we get
\begin{equation}\label{eq_lock4}
\begin{aligned}
	&~P\left( \left. \underset{t\in I_m}{\sup} \left| \bar{x}(t) - \tilde{x}^{\tilde{n}(m)}(t)\right| > \delta \right|\right. \\
	&~~~~~~\left.\underset{t\in I_i}{\sup} \left| \bar{x}(t) - \tilde{x}^{\tilde{n}(i)}(t)\right| \le \delta, 0 \le i < m \right) \\
	{\le} & P\left( \left. \underset{\tilde{n}(m)\le k \le \tilde{n}(m+1)}{\sup}\left\|\boldsymbol{\xi}_{k} - \boldsymbol{\xi}_{\tilde{n}(m)}\right\|_2 > \frac{\delta}{2C_e} \right| \right. \\
	&~~~~~~\left.  \underset{t\in I_i}{\sup} \left| \bar{x}(t) - \tilde{x}^{\tilde{n}(i)}(t)\right| \le \delta, 0 \le i < m \right) \\
	\overset{(a)}{=} &~P\left( \underset{\tilde{n}(m)\le k \le \tilde{n}(m+1)}{\sup}\left\|\boldsymbol{\xi}_{k} - \boldsymbol{\xi}_{\tilde{n}(m)}\right\|_2 > \frac{\delta}{2C_e} \right),
\end{aligned}
\end{equation}
where step $(a)$ is due to the independence of noise, i.e., $ \boldsymbol{\xi}_{k}-\boldsymbol{\xi}_{\tilde{n}(m)} , \tilde{n}(m) \le k \le \tilde{n}(m+1)$ are independent of $\hat{x}_n, 0 \le n \le \tilde{n}(m)$.

%\begin{lemma}[Lemma 4.2 in \cite{borkar2008stochastic}]\label{le_42}
%
%\end{lemma}

The lower bound of the probability that the sequence $\{\hat{x}_n: n \ge 0\}$ remains in the invariant set $\mathcal{I}$ is given by
\begin{align}\label{eq_p_invariant}
	& P\left( \hat{x}_n \in \mathcal{I}, \forall n \ge 0 \right) \nonumber\\
	\overset{(a)}{\ge} & P\left( \underset{t\in I_m}{\sup} \left| \bar{x}(t) - \tilde{x}^{\tilde{n}(m)}(t)\right| \le \delta, \forall m \ge 0 \right)  \nonumber\\
	\overset{(b)}{\ge} & 1 - \sum_{m\ge 0} P\left( \left. \underset{t\in I_m}{\sup} \left| \bar{x}(t) - \tilde{x}^{\tilde{n}(m)}(t)\right| > \delta \right| \right.  \\
	&~~~~~~~~~~~~~~~\left. \underset{t\in I_i}{\sup} \left| \bar{x}(t) - \tilde{x}^{\tilde{n}(i)}(t)\right| \le \delta, 0 \le i < m \right)  \nonumber\\
		\overset{(c)}{\ge} & 1 - \sum_{m\ge 0} P\Bigg( \underset{\tilde{n}(m)\le k \le \tilde{n}(m+1)}{\sup}\left\|\boldsymbol{\xi}_{k} - \boldsymbol{\xi}_{\tilde{n}(m)}\right\|_2 > \frac{\delta}{2C_e} \Bigg),\nonumber
\end{align}
where step $(a)$ is due to Lemma \ref{le_sufficient}, step $(b)$ is due to Lemma 4.2 in \cite{borkar2008stochastic}, and step $(c)$ is due to \eqref{eq_lock4}. Let $\left\|\cdot\right\|_{\infty}$ denote the max-norm, i.e., $\left\|\mathbf{u}\right\|_{\infty} = \max_{l} |[\mathbf{u}]_l|$. Note that for $\mathbf{u} \in \mathbb{R}^{D}$, $\left\|\mathbf{u}\right\|_{2} \le \sqrt{D} \left\|\mathbf{u}\right\|_{\infty}$. Hence we have
\begin{align}\label{eq_lock6-0}
	&~ P\left( \underset{\tilde{n}(m)\le k \le \tilde{n}(m+1)}{\sup}\left\|\boldsymbol{\xi}_{k} - \boldsymbol{\xi}_{\tilde{n}(m)}\right\|_2 > \frac{\delta}{2C_e} \right) \nonumber\\
	\le &~ P\left( \underset{\tilde{n}(m)\le k \le \tilde{n}(m+1)}{\sup} \left\|\boldsymbol{\xi}_{k} - \boldsymbol{\xi}_{\tilde{n}(m)}\right\|_{\infty} > \frac{\delta}{2\sqrt{3}C_e} \right)  \\
	= &~ P\left( \underset{\tilde{n}(m)\le k \le \tilde{n}(m+1)}{\sup} \max_{1 \le l \le 3} \left|\big[\boldsymbol{\xi}_{k}\big]_l - \big[\boldsymbol{\xi}_{\tilde{n}(m)}\big]_l \right| > \frac{\delta}{2\sqrt{3}C_e} \right) \nonumber\\
	= &~ P\left( \max_{1 \le l \le 3} \underset{\tilde{n}(m)\le k \le \tilde{n}(m+1)}{\sup} \left|\big[\boldsymbol{\xi}_{k}\big]_l - \big[\boldsymbol{\xi}_{\tilde{n}(m)}\big]_l \right| > \frac{\delta}{2\sqrt{3}C_e} \right) \nonumber\\
	\le &~ \sum_{l = 1}^3 P\left( \underset{\tilde{n}(m)\le k \le \tilde{n}(m+1)}{\sup} \left|\big[\boldsymbol{\xi}_{k}\big]_l - \big[\boldsymbol{\xi}_{\tilde{n}(m)}\big]_l \right| > \frac{\delta}{2\sqrt{3}C_e} \right). \nonumber
\end{align}

With the increasing $\sigma$-fields $\{\!\mathcal{G}_n\!:\!n\!\ge\!0\!\}$ defined in Appendix \ref{proof:convergence}, we have for $n \ge 0$,
\begin{itemize}
\item[1)] $\boldsymbol{\xi}_n \!=\! \sum_{m=1}^{n} a_{m} \hat{\mathbf{z}}_{m} \sim \mathcal{N}(0, \sum_{m=1}^n a_m^2 \mathbf{I}(\hat{\boldsymbol{\psi}}_{m\!-\!1},\!\mathbf{W}_m)^{-1})$,

\item[2)] $\boldsymbol{\xi}_n$ is $\mathcal{G}_n$-measurable, i.e., $\mathbb{E} \left[ \left. \boldsymbol{\xi}_n \right| \mathcal{G}_n \right] = \boldsymbol{\xi}_n$,

\item[3)] $\mathbb{E} \left[ \left\| \boldsymbol{\xi}_n \right\|^2_2 \right] = \sum_{m=1}^n a_m^2 \operatorname{tr}\Big(\mathbf{I}(\hat{\boldsymbol{\psi}}_{m\!-\!1},\!\mathbf{W}_m)^{-1}\Big)< \infty$,

\item[4)] $\mathbb{E} \left[ \left. \boldsymbol{\xi}_n \right| \mathcal{G}_m \right] = \boldsymbol{\xi}_m$ for all $0 \le m < n$.
\end{itemize}
Therefore, $\left[\boldsymbol{\xi}_n\right]_l, l = 1,2,3$ is a Gaussian martingale with respect to $\mathcal{G}_n$, and satisfies
\begin{align}\label{eq_lock6-1}
\operatorname{Var}\left[\big[\boldsymbol{\xi}_{n+m}\big]_l - \big[\boldsymbol{\xi}_{n}\big]_l\right] = &~\sum_{i = n+1}^{n+m} a_i^2 \left[\mathbf{I}(\hat{\boldsymbol{\psi}}_{i\!-\!1},\!\mathbf{W}_i)^{-1}\right]_{l,l} \nonumber\\
\le &~\sum_{i = n+1}^{n+m} a_i^2 \frac{C_{\mathbf{I}}\sigma_0^2}{|{s}|^2} \\
= &~\frac{C_{\mathbf{I}}\sigma_0^2}{|{s}|^2} \big[b(n) - b(n+m)\big],\nonumber
\end{align}
where $C_{\mathbf{I}} \!\overset{\Delta}{=}\! \max_{l} \max_{i \ge 1} \frac{|{s}|^2}{\sigma_0^2}\big[\mathbf{I}(\hat{\boldsymbol{\psi}}_{i\!-\!1},\!\mathbf{W}_i)^{-1}\big]_{l,l}$. Let $\eta \!=\! \frac{\delta}{2\sqrt{3}C_e}$, $M_i \!=\! \big[\boldsymbol{\xi}_{\tilde{n}(m)+i}\big]_l - \big[\boldsymbol{\xi}_{\tilde{n}(m)}\big]_l, l \!=\! 1, 2, 3$ and $k = {\tilde{n}(m+1) - \tilde{n}(m)}$ in Lemma \ref{le_cheb}, then from \eqref{eq_lock6-0} and \eqref{eq_lock6-1}, we can obtain
\begin{align}\label{eq_lock6}
	%P\left( \left. \underset{t\in I_m}{\sup} \left| \bar{x}(t) - x^n(t)\right| > \delta \right|   \underset{t\in I_m}{\sup} \left| \bar{x}(t) - x^m(t)\right| \le \delta, 0 \le m < n \right)
	&~ P\left( \underset{\tilde{n}(m)\le k \le \tilde{n}(m+1)}{\sup} \left|\big[\boldsymbol{\xi}_{k}\big]_l - \big[\boldsymbol{\xi}_{\tilde{n}(m)}\big]_l \right| > \frac{\delta}{2\sqrt{3}C_e} \right) \nonumber \\
	\le & ~2\exp\left\{-\frac{\delta^2}{24C_e^2\operatorname{Var}\left[\big[\boldsymbol{\xi}_{\tilde{n}(m)+i}\big]_l - \big[\boldsymbol{\xi}_{\tilde{n}(m)}\big]_l\right]}\right\} \\
	\le & ~2\exp\left\{-\frac{\delta^2|{s}|^2}{24C_{\mathbf{I}}C_e^2\big[b(\tilde{n}(m)) - b(\tilde{n}(m+1))\big]\sigma_0^2}\right\}.\nonumber
\end{align}
Combining \eqref{eq_p_invariant}, \eqref{eq_lock6-0} and \eqref{eq_lock6}, we have
%\begin{equation}
\begin{align}\label{eq_lock7}
&~P\left( \hat{x}_n \in \mathcal{I}, \forall n \ge 0 \right) \\
	\ge &~1 -  6\sum_{m \ge 0} \exp\left\{-\frac{\delta^2|{s}|^2}{24C_{\mathbf{I}}C_e^2\big[b(\tilde{n}(m)) - b(\tilde{n}(m+1))\big]\sigma_0^2}\right\}. \nonumber
%	& P\left( \underset{t\in I_m}{\sup} \left| \bar{x}(t) - x^n(t)\right| > \delta \ \text{for some}\  n \ge 0 \right) \\
%	\le & \sum_{n\ge 0}P\left( \left. \underset{t\in I_m}{\sup} \left| \bar{x}(t) - x^n(t)\right| > \delta \right|   \underset{t\in I_m}{\sup} \left| \bar{x}(t) - x^m(t)\right| \le \delta, 0 \le m < n \right) \le 2 \sum_{n\ge 0} e^{-\frac{\rho\delta_1^2}{b(\tilde{n}(m)) - b(\tilde{n}(m+1))}}.
\end{align}
%\end{equation}

To use Lemma \ref{le_sum}, we assume that the step-size $a_n$ satisfies
\begin{equation}\label{eq_lock_constr2}
b(0) = \sum_{i > 0} a_i^2 \le \frac{\delta^2|{s}|^2}{24C_{\mathbf{I}}C_e^2\sigma_0^2}.
\end{equation}
Then, from Lemma \ref{le_sum}, we can obtain
\begin{equation*}%\label{eq_lock8}
\begin{aligned}
	&~\frac{\exp\left\{-\frac{\delta^2|{s}|^2}{24C_{\mathbf{I}}C_e^2\big[b(\tilde{n}(m)) - b(\tilde{n}(m+1))\big]\sigma_0^2}\right\}}{b(\tilde{n}(m)) - b(\tilde{n}(m+1))} \\
	\le &~\frac{\exp\left\{-\frac{\delta^2|{s}|^2}{24C_{\mathbf{I}}C_e^2 b(0)\sigma_0^2}\right\} }{b(0)},
\end{aligned}
\end{equation*}
for $b(\tilde{n}(m)) - b(\tilde{n}(m+1)) < b(\tilde{n}(m)) \le b(0)$. Hence, we have
%\begin{equation*}%\label{eq_lock9}
%\begin{aligned}
%	&~\exp\left\{-\frac{\rho\delta^2}{4C_e^2\big[b(\tilde{n}(m)) - b(\tilde{n}(m+1))\big]}\right\} \\
%	= &~\big[ b(\tilde{n}(m)) - b(\tilde{n}(m+1)) \big] \\
%	&~\cdot \frac{\exp\left\{-\frac{\rho\delta^2}{4C_e^2\big[b(\tilde{n}(m)) - b(\tilde{n}(m+1))\big]}\right\}}{b(\tilde{n}(m)) - b(\tilde{n}(m+1))} \\
%	\le &~\big[ b(\tilde{n}(m)) - b(\tilde{n}(m+1)) \big] \cdot \frac{\exp\left\{-\frac{\rho\delta^2}{4C_e^2b(0)}\right\} }{b(0)},
%\end{aligned}
%\end{equation*}
%and
\begin{align}\label{eq_lock7-2}
	&\sum_{m\ge 0} \exp\left\{-\frac{\delta^2|{s}|^2}{24C_{\mathbf{I}}C_e^2\big[b(\tilde{n}(m)) - b(\tilde{n}(m+1))\big]\sigma_0^2}\right\} \\
	\le &\sum_{m \ge 0} \left[ b(\tilde{n}(m)) - b(\tilde{n}(m+1))\right] \cdot \frac{\exp\left\{-\frac{\delta^2|{s}|^2}{24C_{\mathbf{I}}C_e^2 b(0)\sigma_0^2}\right\}}{b(0)} \nonumber \\
	= &b(0) \cdot \frac{\exp\left\{-\frac{\delta^2|{s}|^2}{24C_{\mathbf{I}}C_e^2b(0)\sigma_0^2}\right\} }{b(0)} = \exp\left\{-\frac{\delta^2|{s}|^2}{24C_{\mathbf{I}}C_e^2b(0)\sigma_0^2}\right\}.\nonumber
	%e^{-\frac{3\rho\delta^2}{8 \pi^2 \alpha^2 e^{2\delta}}},
\end{align}
As $C_e = e^{L (T+a_1)}$, $b(0) = \sum_{i > 0} a_{i}^2$, and $a_n, T, L$ are given by \eqref{eq:stepsize}, \eqref{eq_T}, \eqref{eq_Lip} separately, we can obtain
\begin{equation}\label{eq_exponential}
\begin{aligned}
\frac{\delta^2|{s}|^2}{24C_{\mathbf{I}}C_e^2b(0)\sigma_0^2} = &~ \frac{\delta^2|{s}|^2}{24C_{\mathbf{I}} e^{2L (T+\frac{\alpha}{N_0+1})} \sigma_0^2\sum\limits_{i \ge 1} \frac{\alpha^2}{(i+N_0)^2}} \\
=&~\frac{\delta^2}{\sum\limits_{i \ge 1} \frac{24C_{\mathbf{I}} e^{2L (T+\frac{\alpha}{N_0+1})}}{(i+N_0)^2}} \!\cdot\!\frac{|{s}|^2}{\alpha^2 \sigma_0^2}.
\end{aligned}
\end{equation}
In \eqref{eq_exponential}, $0 < \delta < \inf_{v \in \partial \mathcal{B}} \left| v - \hat{x}_\text{b} \right|$, (\ref{eq_lock_constr}) and (\ref{eq_lock_constr2}) should be satisfied, where a sufficiently large $N_0 \ge 0$ can make both (\ref{eq_lock_constr}) and (\ref{eq_lock_constr2}) true.

To ensures that $\hat{x}_0 + a_{1}\left[\mathbf{f}\left(\hat{\boldsymbol{\psi}}_0, \boldsymbol{\psi}\right)\right]_3$ does not exceed the mainlobe $\mathcal{B}(x)$, i.e., the first step-size $a_{1}$ satisfies
\begin{align*}\left|\hat{x}_0 + a_{1}\left[\mathbf{f}\left(\hat{\boldsymbol{\psi}}_0, \boldsymbol{\psi}\right)\right]_3 - x\right| < \frac{\lambda}{Md},\end{align*}
we can obtain the maximum $\alpha$ as follows
\begin{align}%\label{eq_alpha_range}
	\alpha_{\max} = \frac{(N_0+1)\left(\frac{\lambda}{Md} - \left|x-\hat{x}_0\right|\right)}{\left|\left[\mathbf{f}\left(\hat{\boldsymbol{\psi}}_0, \boldsymbol{\psi}\right)\right]_3\right|} .\nonumber
\end{align}
Hence, from \eqref{eq_exponential}, we have
\begin{align}\label{eq_exponential-2}
\frac{\delta^2|{s}|^2}{24C_{\mathbf{I}}C_e^2b(0)\sigma_0^2} \!\cdot\! \frac{\alpha^2\sigma_0^2}{|{s}|^2}  \ge \frac{\delta^2}{\sum\limits_{i \ge 1} \frac{24C_{\mathbf{I}} e^{2L (T+\frac{\alpha_{\max}}{N_0+1})}}{(i+N_0)^2}} \overset{\Delta}{=} C.
\end{align}

Combining \eqref{eq_lock7}, \eqref{eq_lock7-2} and \eqref{eq_exponential-2}, yields
\begin{equation*}
\begin{aligned}
P\left( \hat{x}_n \in \mathcal{I}, \forall n \ge 0 \right) \ge 1 - 6e^{-\frac{C|{s}|^2}{\alpha^2\sigma_0^2}},
\end{aligned}
\end{equation*}
which completes the proof.

\section{Proof of Lemma \ref{le_gronwall}}\label{sec_proof_le_gronwall}
Apply the discrete Gronwall inequality \cite{holte2009discrete}, leading (\ref{eq_gronwall1}) to
\begin{equation}\label{eq_gronwall3}
\begin{aligned}
	\left\| \bar{\boldsymbol{\psi}}(t_{n+m}) - \tilde{\boldsymbol{\psi}}^n(t_{n+m})\right\|_2 \le C e^{L\sum_{i=1}^m a_{n+i}}.
\end{aligned}
\end{equation}
Since $1 \le m \le n_T$ and $n_T = \inf \left\{i \in \mathbb{Z}: t_{n+i} \ge t_n + T \right\} $, we get
\begin{equation}\label{eq_gronwall33}
\begin{aligned}
	\sum_{i=1}^m a_{n+i} = t_{n+m} - t_n \le T + a_{n+n_T} \le T + a_1.
\end{aligned}
\end{equation}
By combining \eqref{eq_gronwall3} and \eqref{eq_gronwall33}, we have
\begin{equation}\label{eq_gronwall4}
\begin{aligned}
	\left\| \bar{\boldsymbol{\psi}}(t_{n+m}) - \tilde{\boldsymbol{\psi}}^n(t_{n+m})\right\|_2 \le &~C e^{L (T+a_1)}.
\end{aligned}
\end{equation}

For $\forall t \in [t_{n+m-1}, t_{n+m}], 1 \le m \le n_T$, from \eqref{eq_continuous}, we have
\begin{equation*}\label{eq_gronwall5}
\begin{aligned}
	\bar{\boldsymbol{\psi}}(t) & = \bar{\boldsymbol{\psi}}(t_{n+m-1}) + \frac{(t-t_{n+m-1})\left[\bar{\boldsymbol{\psi}}(t_{n+m}) - \bar{\boldsymbol{\psi}}(t_{n+m-1})\right]}{a_{n+m}}\\
	& = \gamma \bar{\boldsymbol{\psi}}(t_{n+m-1}) + (1 - \gamma) \bar{\boldsymbol{\psi}}(t_{n+m}),
\end{aligned}
\end{equation*}
where $\gamma = \frac{t_{n+m} - t}{a_{n+m}} \in [0, 1]$. Then, we can get \eqref{eq_gronwall6}, where step $(a)$ is according to the definition of $\tilde{\boldsymbol{\psi}}^n(t)$ in \eqref{eq_ODE_new}, step $(b)$ is due to (\ref{eq_gronwall4}), step $(c)$ is obtained from (\ref{eq_CT}), and step $(d)$ is obtained by using $\gamma = \frac{t_{n+m} - t}{a_{n+m}}$.
\begin{figure*}[ht]
\begin{align}\label{eq_gronwall6}
	&~\left\| \bar{\boldsymbol{\psi}}(t) - \tilde{\boldsymbol{\psi}}^n(t)\right\|_2 \\
	= &~\left\| \gamma\left(\bar{\boldsymbol{\psi}}(t_{n+m-1}) - \tilde{\boldsymbol{\psi}}^n(t)\right) + (1 - \gamma) \left(\bar{\boldsymbol{\psi}}(t_{n+m}) - \tilde{\boldsymbol{\psi}}^n(t)\right) \right\|_2 \nonumber\\
	\overset{(a)}{=} &~\left\| \gamma \left[ \bar{\boldsymbol{\psi}}(t_{n+m-1}) - \tilde{\boldsymbol{\psi}}^n(t_{n+m-1}) - \int_{t_{n+m-1}}^t \mathbf{f}\left(\tilde{\boldsymbol{\psi}}^n(s), \boldsymbol{\psi}\right) ds \right] + (1 - \gamma) \left[ \bar{\boldsymbol{\psi}}(t_{n+m}) - \tilde{\boldsymbol{\psi}}^n(t_{n+m}) -  \int_{t_{n+m}}^t \mathbf{f}\left(\tilde{\boldsymbol{\psi}}^n(s), \boldsymbol{\psi}\right) ds \right] \right\|_2 \nonumber\\
	\le &~\gamma\left\| \int_{t_{n+m-1}}^t \mathbf{f}\left(\tilde{\boldsymbol{\psi}}^n(s), \boldsymbol{\psi}\right) ds \right\|_2 + (1 - \gamma)\left\| \int_{t_{n+m}}^t \mathbf{f}\left(\tilde{\boldsymbol{\psi}}^n(s), \boldsymbol{\psi}\right) ds \right\|_2 + \gamma\left\| \bar{\boldsymbol{\psi}}(t_{n+m-1}) - \tilde{\boldsymbol{\psi}}^n(t_{n+m-1}) \right\|_2 \nonumber \\
	&~+ (1 - \gamma)\left\| \bar{\boldsymbol{\psi}}(t_{n+m}) - \tilde{\boldsymbol{\psi}}^n(t_{n+m}) \right\|_2 \nonumber\\
	\overset{(b)}{\le} &~\gamma  \int_{t_{n+m-1}}^t \left\|\mathbf{f}\left(\tilde{\boldsymbol{\psi}}^n(s), \boldsymbol{\psi}\right) \right\|_2 ds + (1 - \gamma)\int_t^{t_{n+m}} \left\| \mathbf{f}\left(\tilde{\boldsymbol{\psi}}^n(s), \boldsymbol{\psi}\right) \right\|_2 ds + Ce^{L (T+a_1)} \nonumber\\
	\overset{(c)}{\le} &~C_{\mathbf{f}} \gamma (t - t_{n+m-1}) + C_{\mathbf{f}}(1 - \gamma)(t_{n+m} - t) + Ce^{L (T+a_1)} \nonumber\\
	\overset{(d)}{\le} &~2C_{\mathbf{f}} a_{n+m}\gamma(1-\gamma)  + Ce^{L (T+a_1)} \le \frac{C_{\mathbf{f}} a_{n+m}}{2} + Ce^{L (T+a_1)} \nonumber\\
	\le &~\underset{1 \le m \le n_T}{\sup} \frac{C_{\mathbf{f}} a_{n+m}}{2} + C e^{L (T+a_1)} = \frac{C_{\mathbf{f}} a_{n+1}}{2} + C e^{L (T+a_1)}.\nonumber
%	& \left| \bar{x}(t) - \tilde{x}^n(t)\right| \\
%	= & \left| \gamma(\bar{x}(t_{n+m-1}) - \tilde{x}^n(t)) + (1 - \gamma) (\bar{x}(t_{n+m}) - \tilde{x}^n(t)) \right| \\
%	\overset{(a)}{=} & \left| \gamma \left[ \bar{x}(t_{n+m-1}) - \tilde{x}^n(t_{n+m-1}) - \int_{t_{n+m-1}}^t f(\tilde{x}^n(s), x) ds \right] \right. \\
%	& \left. + (1 - \gamma) \left[ \bar{x}(t_{n+m}) - \tilde{x}^n(t_{n+m}) -  \int_{t_{n+m}}^t f(\tilde{x}^n(s), x) ds \right] \right| \\
%	\le & \gamma\left| \int_{t_{n+m-1}}^t f(\tilde{x}^n(s), x) ds \right| + (1 - \gamma)\left| \int_{t_{n+m}}^t f(\tilde{x}^n(s), x) ds \right| \\
%	& + \gamma\left| \bar{x}(t_{n+m-1}) - \tilde{x}^n(t_{n+m-1}) \right| \\&+ (1 - \gamma)\left| \bar{x}(t_{n+m}) - \tilde{x}^n(t_{n+m}) \right| \\
%	\overset{(b)}{\le} & \gamma  \int_{t_{n+m-1}}^t \left|f(\tilde{x}^n(s), x) \right| ds + (1 - \gamma)\int_t^{t_{n+m}} \left| f(\tilde{x}^n(s), x) \right| ds \\
%	& + Ce^{L (T+a_1)} \\
%	\overset{(c)}{\le} & \sqrt{M} \gamma (t - t_{n+m-1}) + \sqrt{M}(1 - \gamma)(t_{n+m} - t) + Ce^{L (T+a_1)} \\
%	\overset{(d)}{\le} & 2\sqrt{M} a_{n+m}\gamma(1-\gamma)  + Ce^{L (T+a_1)} \\
%	\le & \frac{\sqrt{M} a_{n+m}}{2} + Ce^{L (T+a_1)} \\
%	\le & \underset{1 \le m \le n_T}{\sup} \frac{\sqrt{M} a_{n+m}}{2} + C e^{L (T+a_1)} \\
%	= & \frac{\sqrt{M} a_{n+1}}{2} + C e^{L (T+a_1)},
\end{align}
\hrulefill
\end{figure*}

Therefore, from \eqref{eq_gronwall6}, we can obtain
\begin{equation*}
\begin{aligned}
	& \underset{t\in\left[ t_n, t_{n+n_T} \right]}{\sup} \left\| \bar{\boldsymbol{\psi}}(t) - \tilde{\boldsymbol{\psi}}^n(t)\right\|_2 	\le \frac{C_{\mathbf{f}} a_{n+1}}{2} + C e^{L (T+a_1)},
\end{aligned}
\end{equation*}
which completes the proof.

\end{document}